\documentclass[11pt,a4paper]{article}

\usepackage[utf8]{inputenc}
\usepackage[T1]{fontenc}
\usepackage[english]{babel}
\usepackage{amsmath,amssymb,amsthm,mathtools}
\usepackage{aliascnt}
\usepackage{bm}
\usepackage{enumitem}
\usepackage{booktabs}
\usepackage{tabularx}
\usepackage{geometry}
\usepackage{graphicx}
\usepackage{xcolor}
\usepackage{url}
\usepackage[numbers,sort&compress]{natbib}
\usepackage[colorlinks=true,allcolors=blue]{hyperref}
\usepackage[capitalize]{cleveref}
\DeclareGraphicsExtensions{.pdf,.png,.jpg}
\graphicspath{{figures/}}
\newcommand{\tr}{\operatorname{tr}}
\newcommand{\E}{\mathbb E}
\newcommand{\Pp}{\mathbb P}
\newcommand{\Var}{\operatorname{Var}}
\newcommand{\Cov}{\operatorname{Cov}}
\newcommand{\Hh}{\mathbb H}
\newcommand{\Z}{\mathbb Z}
\newcommand{\dist}{\operatorname{dist}}
\newcommand{\osc}{\operatorname{osc}}
\newcommand{\TV}{\operatorname{TV}}
\newcommand{\SAW}{\mathrm{SAW}}
\newcommand{\free}{\mathrm{free}}
\newcommand{\fixb}{\mathrm{fix}}
\renewcommand{\d}{\operatorname{d}\!}

\theoremstyle{plain}
\newtheorem{theorem}{Theorem}[section]
\newaliascnt{proposition}{theorem}
\newtheorem{proposition}[proposition]{Proposition}
\aliascntresetthe{proposition}
\newaliascnt{lemma}{theorem}
\newtheorem{lemma}[lemma]{Lemma}
\aliascntresetthe{lemma}
\newaliascnt{corollary}{theorem}
\newtheorem{corollary}[corollary]{Corollary}
\aliascntresetthe{corollary}
\theoremstyle{definition}
\newaliascnt{definition}{theorem}
\newtheorem{definition}[definition]{Definition}
\aliascntresetthe{definition}
\newaliascnt{problem}{theorem}
\newtheorem{problem}[problem]{Open Problem}
\aliascntresetthe{problem}
\theoremstyle{remark}
\newaliascnt{remark}{theorem}

\aliascntresetthe{remark}

\title{\Large\bf Boundary Free Energies, Quenched Mixing, and Gibbs-State Selection in Disordered Ising Models}
\author{Mauris Chueng\thanks{These authors contributed equally to this work.}\ \textsuperscript{,}\thanks{Corresponding author: maurischueng@gmail.com} \and Hexiang Wang\footnotemark[1] \and Keheng Zhu\footnotemark[1]}
\date{}

\newcommand{\Addresses}{{%
		\bigskip
		\footnotesize
		
		\textsc{Mauris Chueng}, \textsc{School of Statistics and Data Science, Jilin University of Finance and Economics, Changchun, 130117, China}\par\nopagebreak
		\texttt{maurischueng@gmail.com}
		\medskip
		
		\textsc{Hexiang Wang}, \textsc{School of Mathematical Sciences, Nankai University, Tianjin, 300071, China}\par\nopagebreak
		\texttt{Kui6539@outlook.com}
		\medskip
		
		\textsc{Keheng Zhu}, \textsc{Academy for Multidisciplinary Studies and School of Mathematical Sciences, Capital Normal University, Beijing, 100048, China}\par\nopagebreak
		\texttt{hexistartop@gmail.com}
		\medskip
}}

\begin{document}
\maketitle

\begin{abstract}
We study boundary free energies and half-space responses in nearest-neighbor disordered Ising models. First, we prove the existence of fixed-depth tangential pressure and identify its derivative almost everywhere with the limiting boundary response. Second, under quenched exponential boundary-response mixing, we prove Gibbs-state selection independence and exponential convergence of finite-depth pressures. We further show that uniform one-spin mixing yields DLR uniqueness and Gaussian normal localization, and verify the required mixing conditions whenever $(2d-1)\mathbb E\tanh(\beta\vert{}J\vert{})<1$. Finally, we establish an all-face surface limit in the bounded Dobrushin regime and provide counterexamples that delimit general surface and stiffness claims.
\end{abstract}

%

\section{Introduction}

Boundary terms in disordered spin systems depend on the observable and on the geometry used to normalize it. A free-to-fixed correction usually has a nonzero quenched mean and contains local boundary noise. A periodic-to-antiperiodic difference instead compares two Hamiltonians in the same disorder by reversing a seam; its mean may vanish by symmetry, and it is a domain-wall observable. These distinctions are essential in disordered Ising models \cite{EdwardsAnderson1975,Chayes1986,FisherHuse1988,NewmanStein2003} and in the stiffness literature \cite{McMillan1984,BrayMoore1987,HartmannYoung2001,Boettcher2005,Amoruso2006}.

The present paper integrates three complementary levels of analysis. The first is finite-volume: comparison, interpolation, and concentration require no thermodynamic state. The second is tangential: at each fixed normal depth, almost-additivity produces an expected pressure for all integrable product coupling laws. The third is normal: comparison with a half-space state requires a spatial estimate that controls the influence of the remote top boundary. Keeping these levels separate prevents finite-volume identities or concentration estimates from being mistaken for low-temperature localization.

Our main synthesis is the derivative representation from concavity at finite inverse temperature. The fixed-depth argument applies without change when the bottom couplings are multiplied by any $r\in[0,1]$. The resulting limiting pressure is concave and uniformly Lipschitz in $r$. Standard secant-slope bounds for pointwise convergent concave functions then identify its derivative, at almost every $r$, with the limit of the finite-volume boundary responses. Thus no separate strip-interpolation hypothesis is needed.

We next formulate boundary-response mixing (BRM), a quenched estimate for one boundary spin. Exact DLR disintegration is performed before taking limits. It implies that the weighted half-space response is independent of a measurable Gibbs-state selection, and it compares that response with the fixed-depth pressure at an explicit exponential rate. The proof does not assume that an arbitrary selector is translation-covariant: translated selectors are constructed and the relevant full-probability events are intersected over the countable group of tangential shifts. A stronger estimate on a countable class of local observables yields uniqueness of the entire DLR state. We prove that it is enough to have a uniform spin-mean estimate at every site and in every reduced finite volume: sequential maximal coupling lifts this one-site condition to local-observable mixing.

For i.i.d. internal couplings, disagreement exploration verifies both levels when
\begin{equation}\label{eq:intro-q}
q_{\beta,d}:=(2d-1)\E\tanh(\beta|J|)<1.
\end{equation}
The underlying comparison is pathwise, before disorder averaging, and is uniform over deterministic or disorder-dependent boundary kernels. This regime includes sufficiently dilute signed couplings even at low temperature, but it excludes the fully occupied Gaussian law as $\beta$ becomes large. Separately, bounded couplings satisfying the strict Dobrushin condition yield an all-face theorem for regular rectangles, including sample convergence under an explicit summability condition. This stronger conclusion is not inferred from one-face BRM.

Finally, one-dimensional examples show that arbitrary van Hove exhaustions can have different expected surface limits and that even regular intervals need not have a sample limit. Seam-flip differences obey a surface-order variance upper bound under explicit conditional sign-flip hypotheses, but their stiffness scale is not universal under symmetry and moment assumptions alone. An independent-chain model separates ordinary boundary central-limit fluctuations from a domain-wall statistic. These limitations isolate, rather than solve, the ordinary fully occupied low-temperature Gaussian problem.

No theorem below is obtained by combining incompatible normalizations or by interpreting one-face response mixing as an all-face comparison principle.

\section{Geometry and finite-volume calculus}

Fix $d\geq2$ and put $k=d-1$. For $L\geq1$ and a tangential rectangle
\begin{equation}\label{eq:slab}
T_{\bm n}=\prod_{i=2}^d\{1,\ldots,n_i\},
\qquad
\Lambda_{L,\bm n}=\{1,\ldots,L\}\times T_{\bm n},
\qquad
|\bm n|=\prod_{i=2}^dn_i,
\end{equation}
let $E_{L,\bm n}$ be the internal nearest-neighbor edges. The bottom crossing edge at $z\in T_{\bm n}$ joins $(1,z)$ to the fixed exterior spin $+1$ at $(0,z)$. Internal couplings $(J_e)$ and bottom couplings $(K_z)$ are independent i.i.d. fields, possibly with different laws. The other faces are free. For $r\in[0,1]$ set
\begin{align}
H_{L,\bm n}^{\free}(\sigma)
 &=-\sum_{e=\{x,y\}\in E_{L,\bm n}}J_e\sigma_x\sigma_y,
\label{eq:H-free}\\
H_{L,\bm n,r}(\sigma)
 &=H_{L,\bm n}^{\free}(\sigma)
   -r\sum_{z\in T_{\bm n}}K_z\sigma_{(1,z)}.
\label{eq:H-r}
\end{align}
At finite $\beta$, write $F_{L,\bm n}(r)=-\beta^{-1}\log Z_{L,\bm n}(r)$ and
\begin{equation}\label{eq:Delta-slab}
\Delta F_{L,\bm n}(r)=F_{L,\bm n}(r)-F_{L,\bm n}(0).
\end{equation}
At $\beta=\infty$, $F$ is replaced by the ground-state energy $E=\min_\sigma H$ and the correction is denoted by $\Delta E$.

For a general finite zero-field Ising Hamiltonian, let $B$ be a set of crossing edges, let $q_e(\sigma)=\sigma_x\psi_y$ for a prescribed exterior spin $\psi_y$, and let $K=(K_e)_{e\in B}$. The next result collects the finite-volume facts used below.

\begin{proposition}[Finite-volume comparison and interpolation]\label{prop:finite-calculus}
For every realization and every $r\in[0,1]$,
\begin{equation}\label{eq:pointwise-comparison}
-r\sum_{e\in B}|K_e|\leq F(r)-F(0)\leq0.
\end{equation}
The same statement holds for ground-state energies. At finite $\beta$, if $\E|K_e|<\infty$, then
\begin{equation}\label{eq:linear-response}
\frac{\d}{\d r}\E F(r)
=-\sum_{e\in B}\E\bigl[K_e\langle q_e\rangle_r\bigr],
\qquad
\left|\frac{\d}{\d r}\E F(r)\right|
\leq\sum_{e\in B}\E|K_e|.
\end{equation}
If $K$ is centered Gaussian with covariance matrix $\Gamma$ and
$H_t=H_0-\sqrt t\sum_{e\in B}K_eq_e$, then
\begin{equation}\label{eq:gaussian-correlated}
\frac{\d}{\d t}\E\log Z_t
=\frac{\beta^2}{2}\sum_{e,b\in B}\Gamma_{eb}
\E\Cov_t(q_e,q_b).
\end{equation}
For independent variance-$v$ Gaussian coordinates this gives
\begin{equation}\label{eq:gaussian-delta}
\E[F(1)-F(0)]
=-\frac{\beta v}{2}\sum_{e\in B}\int_0^1
\E\bigl[1-\langle q_e\rangle_t^2\bigr]\d t.
\end{equation}
For independent centered coordinates of common finite variance $v$, the same sign and the bound
$-\beta v|B|/2\leq\E[F(1)-F(0)]\leq0$ follow from the independent-copy identity.
For all averaged interpolation and independent-copy assertions, $K$ is independent of any randomness entering $H_0$; equivalently, the calculations are first made conditional on $H_0$ and then averaged.
\end{proposition}

\begin{proof}
The Gibbs measure for $H_0$ is invariant under the global spin flip. Hence
\begin{equation}\label{eq:cosh-ratio}
\frac{Z(r)}{Z(0)}
=\left\langle\exp\left(\beta r\sum_{e\in B}K_eq_e\right)\right\rangle_0
=\left\langle\cosh\left(\beta r\sum_{e\in B}K_eq_e\right)\right\rangle_0.
\end{equation}
The pointwise bounds follow from
\begin{align}
1
&\leq \frac{Z(r)}{Z(0)}
\leq \exp\left(\beta r\sum_{e\in B}|K_e|\right),
\notag\\
-r\sum_{e\in B}|K_e|
&\leq -\frac1\beta\log\frac{Z(r)}{Z(0)}
=F(r)-F(0)
\leq0.
\label{eq:comparison-derivation}
\end{align}
If $\sigma^0$ minimizes $H_0$, one of $\sigma^0$ and $-\sigma^0$ has nonnegative boundary contribution. Therefore the zero-temperature bounds are
\begin{align*}
E(r)
&\leq \min\{H_r(\sigma^0),H_r(-\sigma^0)\}
\leq E(0),
\\
E(r)
&=\min_\sigma\left(H_0(\sigma)-r\sum_{e\in B}K_eq_e(\sigma)\right)
\geq E(0)-r\sum_{e\in B}|K_e|.
\end{align*}

At finite temperature, differentiation and $|\langle q_e\rangle_r|\leq1$ give
\begin{align*}
\frac{\d}{\d r}F(r)
=\left\langle\frac{\partial H_r}{\partial r}\right\rangle_r
&=-\sum_{e\in B}K_e\langle q_e\rangle_r,
\\
\left|\frac{\d}{\d r}\E F(r)\right|
&\leq\sum_{e\in B}\E|K_e|.
\end{align*}
For the square-root Gaussian path, Gaussian integration by parts yields the full chain
\begin{align}
\frac{\d}{\d t}\E\log Z_t
&=\frac{\beta}{2\sqrt t}\sum_{e\in B}
\E[K_e\langle q_e\rangle_t]
\\
&=\frac{\beta}{2\sqrt t}\sum_{e,b\in B}\Gamma_{eb}
\E\left[\frac{\partial}{\partial K_b}\langle q_e\rangle_t\right]
\\
&=\frac{\beta^2}{2}\sum_{e,b\in B}\Gamma_{eb}
\E\Cov_t(q_e,q_b).
\label{eq:Gaussian-IBP-chain}
\end{align}
Writing $C_t=(\Cov_t(q_e,q_b))_{e,b\in B}$, both $\Gamma$ and $C_t$ are positive semidefinite, and hence
\begin{align*}
\sum_{e,b\in B}\Gamma_{eb}\Cov_t(q_e,q_b)
=\tr(\Gamma C_t)
=\tr(\Gamma^{1/2}C_t\Gamma^{1/2})
\geq0.
\end{align*}
Moreover,
\begin{align*}
|\log Z_t-\log Z_0|
&\leq\beta\sqrt t\sum_{e\in B}|K_e|,
\\
\left|\frac{\beta}{2\sqrt t}\sum_{e\in B}
\E[K_e\langle q_e\rangle_t]\right|
&\leq\frac{\beta}{2\sqrt t}\sum_{e\in B}\E|K_e|,
\end{align*}
so the derivative is integrable at $t=0$. For independent variance-$v$ coordinates, $\Gamma=vI$ and $q_e^2=1$, and therefore
\begin{align*}
\E[F(1)-F(0)]
&=-\frac1\beta\int_0^1\frac{\d}{\d t}\E\log Z_t\,\d t
\\
&=-\frac{\beta v}{2}\sum_{e\in B}\int_0^1
\E\left[1-\langle q_e\rangle_t^2\right]\d t.
\end{align*}

For the non-Gaussian assertion, replace $K_e$ by an independent copy $K_e'$ one coordinate at a time. Conditional on all other variables, set $m_e(x)=\langle q_e\rangle$ when $K_e=x$. Centering, exchangeability, and
$m_e'(x)=\beta r(1-m_e(x)^2)$ give
\begin{align}
\E[K_em_e(K_e)]
&=\frac12\E\left[(K_e-K_e')
\bigl(m_e(K_e)-m_e(K_e')\bigr)\right]
\notag\\
&=\frac12\E\left[(K_e-K_e')^2
\int_0^1m_e'\bigl(K_e'+u(K_e-K_e')\bigr)\d u\right]
\notag\\
&=\frac{\beta r}{2}\E\left[(K_e-K_e')^2
\int_0^1\bigl(1-m_{e,u}^2\bigr)\d u\right].
\label{eq:copy-identity}
\end{align}
Since $\E(K_e-K_e')^2=2v$,
\begin{align*}
0
&\leq\E[K_em_e(K_e)]
\leq\beta rv,
\\
-\frac{\beta v}{2}|B|
&=-\sum_{e\in B}\int_0^1\beta rv\,\d r
\leq\E[F(1)-F(0)]
\leq0.
\end{align*}
These are finite-volume identities.
\end{proof}

\subsection{Endpoint-correct interpolation and the Gaussian specialization}

The endpoint of an interpolation is part of the statement, not a cosmetic choice. For example, with $\mu=\E|K|$ the path
\begin{equation}\label{eq:incorrect-path}
h_e(t)=\sqrt t\,K_e+(1-\sqrt t)\mu
=\mu+\sqrt t(K_e-\mu)
\end{equation}
starts from a deterministic boundary field of strength $\mu$ and ends at the random boundary field. It does not compare the free and fixed Hamiltonians. If the $K_e$ are independent centered Gaussians of variance $v$, differentiation gives
\begin{equation}\label{eq:incorrect-derivative}
\frac{\d}{\d t}\E\log Z_t
=\frac{\beta^2v}{2}\sum_{e\in B}\E(1-m_e(t)^2)
-\frac{\beta\mu}{2\sqrt t}\sum_{e\in B}\E m_e(t),
\end{equation}
where $m_e(t)=\langle q_e\rangle_t$. The last term is integrable but may have neither a finite right limit nor a fixed sign because $m_e(0)$ need not vanish. The square-root path used in \cref{prop:finite-calculus} has the required endpoints and retains all covariance cross terms when the Gaussian field is correlated.

For independent Gaussian boundary coordinates, it is sometimes more convenient to use the linear amplitude $r$ from \cref{eq:H-r}. Finite-volume Gaussian integration by parts then yields
\begin{equation}\label{eq:gaussian-linear-response}
\frac1{|\bm n|}\E\sum_{z\in T_{\bm n}}K_z
\langle\sigma_{(1,z)}\rangle_{L,\bm n,r}
=\beta vr\,\frac1{|\bm n|}\sum_{z\in T_{\bm n}}
\E\left[1-\langle\sigma_{(1,z)}\rangle_{L,\bm n,r}^2\right].
\end{equation}
This is a finite-volume identity. Passing either side to an infinite-depth state is a separate localization problem.

\begin{proposition}[Gaussian concentration]\label{prop:gaussian-concentration}
Suppose that all internal and boundary coordinates in a finite region $\Lambda$ are independent centered Gaussian variables of variance $v$. If $E(\Lambda)$ and $B(\Lambda)$ are its internal and boundary edge sets, then
\begin{align}
\Var(\Delta F_\Lambda)&\leq v\bigl(4|E(\Lambda)|+|B(\Lambda)|\bigr),
\label{eq:generic-var}\\
\Pp\bigl(|\Delta F_\Lambda-\E\Delta F_\Lambda|\geq u\bigr)
&\leq2\exp\left\{-\frac{u^2}{2v(4|E(\Lambda)|+|B(\Lambda)|)}\right\}.
\label{eq:generic-tail}
\end{align}
The same tail bound holds for the ground-state correction.
\end{proposition}

\begin{proof}
Put $G=\Delta F_\Lambda$. For an internal edge $a=\{x,y\}$ and a boundary edge $e$,
\begin{align*}
\frac{\partial G}{\partial J_a}
=-\langle\sigma_x\sigma_y\rangle_{\fixb}
+\langle\sigma_x\sigma_y\rangle_{\free},\quad
\left|\frac{\partial G}{\partial J_a}\right|\leq2&,
\\
\frac{\partial G}{\partial K_e}
=-\langle q_e\rangle_{\fixb},\quad
\left|\frac{\partial G}{\partial K_e}\right|\leq1&.
\end{align*}
Consequently,
\begin{align*}
\|\nabla G\|_2^2
&\leq4|E(\Lambda)|+|B(\Lambda)|,
\\
\Var(G)
&\leq v\E\|\nabla G\|_2^2
\leq v\bigl(4|E(\Lambda)|+|B(\Lambda)|\bigr),
\\
\Pp(|G-\E G|\geq u)
&\leq2\exp\left\{-\frac{u^2}
{2v(4|E(\Lambda)|+|B(\Lambda)|)}\right\},
\end{align*}
by Gaussian Poincar\'e and concentration \cite{Ledoux2001,BoucheronLugosiMassart2013}. Minima of affine functions have the same Lipschitz constants. For cubes,
\begin{align*}
\Var\left(\frac{\Delta F_L}{L^{d-1}}\right)
\leq\frac{v(4|E(\Lambda_L)|+|B(\Lambda_L)|)}{L^{2(d-1)}}
=O(L^{2-d}).
\end{align*}
Thus the estimate transfers convergence of normalized means to sample convergence in $d\geq3$, but it neither proves convergence of those means nor gives surface-order self-averaging in $d=2$.
\end{proof}

\section{Normalization and fluctuation scales}

For the cube $\Lambda_n=\{1,\ldots,n\}^d$, let $b_n=2dn^{d-1}$ be the number of oriented crossing edges, let $f_n^b=F_n^b/n^d$, and let $\Delta F_n=F_n^{\fixb}-F_n^{\free}$. Then, for sample or expected free energies,
\begin{equation}\label{eq:normalization}
n(f_n^{\fixb}-f_n^{\free})
=\frac{\Delta F_n}{n^{d-1}}
=2d\,\frac{\Delta F_n}{b_n}.
\end{equation}
Thus the normalization by $n^{d-1}$ differs from normalization by crossing edges by the factor $2d$. The following objects should not be conflated.

\begin{center}
\small
\begin{tabularx}{.96\textwidth}{@{}>{\raggedright\arraybackslash}p{.27\textwidth}>{\raggedright\arraybackslash}p{.20\textwidth}X@{}}
\toprule
Object & Status & Interpretation \\
\midrule
$\E\Delta F_n$ & deterministic & Mean free-to-fixed correction; it may be of order $n^{d-1}$. \\
$\Delta F_n-\E\Delta F_n$ & random & Centered fluctuation containing local boundary noise and bulk-mediated response. \\
$D_n=F_n^{\rm tw}-F_n^{\rm untw}$ & random & Difference across a flipped seam in the same sample; symmetry may cancel the conditional mean. \\
$\Var D_n$, $\E|D_n|$, or a quantile & deterministic & Inequivalent notions of typical stiffness without further moment control. \\
\bottomrule
\end{tabularx}
\end{center}

If $(\E D_n^2)^{1/2}=n^{\theta_{\rm rms}+o(1)}$, then
\begin{equation}\label{eq:theta-density}
\left\|\frac{D_n}{b_n}\right\|_2
=n^{\theta_{\rm rms}-(d-1)+o(1)}.
\end{equation}
A mean boundary term of order $b_n$ and a smaller centered fluctuation may therefore coexist. Moreover, \cref{eq:pointwise-comparison} gives, for $p\geq1$,
\begin{equation}\label{eq:UI}
\E\left|\frac{\Delta F_n}{n^{d-1}}\right|^p
\leq(2d)^p\E|K|^p.
\end{equation}
For merely integrable $K$, the boundary averages are uniformly integrable by the de la Vall\'ee Poussin criterion. Hence convergence in probability may be passed to the mean when the stated moment hypotheses supply uniform integrability; an $O(n^{d-1})$ comparison alone does not imply a limit.

Under the Gaussian assumptions of \cref{prop:gaussian-concentration}, the pointwise comparison and the Poincar\'e bound give
\begin{equation}\label{eq:generic-fluctuation-scale}
\|\Delta F_n-\E\Delta F_n\|_2
=O\left(\min\{n^{d/2},n^{d-1}\}\right).
\end{equation}
The $n^{d/2}$ term is the generic bulk-coordinate estimate; the $n^{d-1}$ term follows from $|\Delta F_n|\leq\sum_{e\in B_n}|K_e|$. This is weaker than the seam estimate in \cref{thm:seam}, because no conditional antisymmetry removes the contribution of internal disorder.

\begin{figure}[t]
\centering
\includegraphics[width=.91\linewidth]{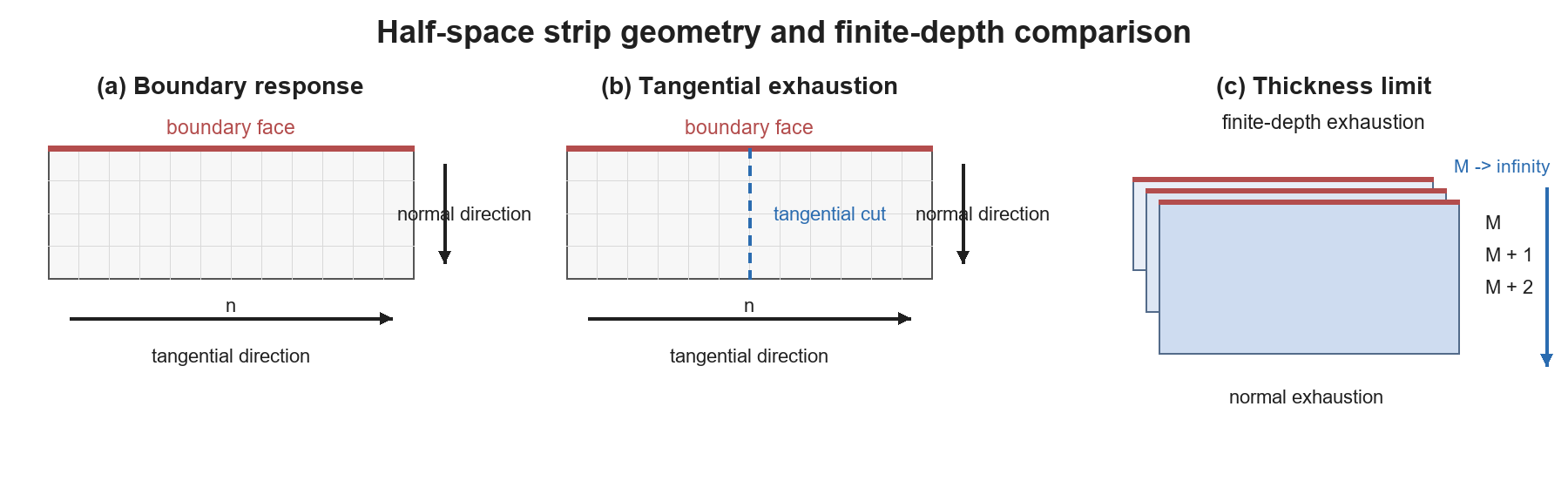}
\caption{Slab geometry. Panel (a) marks the boundary bonds of amplitude $r$; panel (b) shows a tangential cut used in almost-additivity; panel (c) shows the normal exhaustion, with the graphic depth $M$ playing the role of the textual depth $L$. All other finite-volume faces are free.}
\label{fig:strip-geometry}
\end{figure}

\section{Fixed-depth tangential pressure}

Assume in this section that
\begin{equation}\label{eq:first-moments}
\mu_{\rm int}:=\E|J_e|<\infty,
\qquad
\mu_{\rm b}:=\E|K_z|<\infty.
\end{equation}
No symmetry of either coupling law is needed. Independence and tangential stationarity are used in the factorization and tiling argument.

\begin{lemma}[Tangential almost-additivity]\label{lem:almost-add}
Suppose that two tangential rectangles with equal transverse side lengths are concatenated in direction $i\in\{2,\ldots,d\}$. For every fixed $L$, $r\in[0,1]$, and $\beta\in(0,\infty]$, the difference between the expected correction of the concatenated slab and the sum of the two expected corrections has absolute value at most $2L\mu_{\rm int}\prod_{\substack{2\leq j\leq d\\j\ne i}}n_j.$.
\end{lemma}

\begin{proof}
Let $I$ be the set of $L\prod_{j\ne i}n_j$ interface bonds and let $F^{\rm cut}(r)$ be the free energy after deleting them. The cut Hamiltonian factorizes, and the one-bond comparison gives
\begin{align*}
\left|\Delta F(r)-\Delta F^{\rm cut}(r)\right|
&\leq |F(r)-F^{\rm cut}(r)|
+|F(0)-F^{\rm cut}(0)|
\\
&\leq2\sum_{e\in I}|J_e|,
\\
\left|\E\Delta F(r)-\E\Delta F_1(r)-\E\Delta F_2(r)\right|
&\leq2|I|\mu_{\rm int}
=2L\mu_{\rm int}\prod_{j\ne i}n_j.
\end{align*}
The same inequalities hold for minima at zero temperature.
\end{proof}

\begin{theorem}[Rectangular fixed-depth pressure]\label{thm:fixed-pressure}
Let $d\geq2$, fix $L\geq1$, and assume \cref{eq:first-moments}. For every $r\in[0,1]$ and every $\beta\in(0,\infty]$, the rectangular tangential limit
\begin{equation}\label{eq:p-L-r}
p_L(r)=\lim_{\min_{2\leq i\leq d}n_i\to\infty}
\frac{1}{|\bm n|}\E\Delta F_{L,\bm n}(r)
\end{equation}
exists and is deterministic, with $\Delta F$ interpreted as $\Delta E$ at $\beta=\infty$. In particular,
\begin{equation}\label{eq:s-L}
s_L(\beta):=p_L(1),
\qquad
-\mu_{\rm b}\leq s_L(\beta)\leq0.
\end{equation}
Only rectangular tangential boxes with all side lengths tending to infinity are asserted.
\end{theorem}

\begin{proof}
Fix a tangential block size $a\geq1$. Inside $T_{\bm n}$ take the tiled core
\begin{equation*}
Q_{\bm n}=\prod_{i=2}^d\{1,\ldots,a\lfloor n_i/a\rfloor\},
\end{equation*}
which is a union of $M_{\bm n}=\prod_{i=2}^d\lfloor n_i/a\rfloor$ translates of the $a^k$ block. Delete all internal bonds between these blocks and all internal bonds from the core to its remainder $R_{\bm n}$. The number of internal seams inside the core is at most $kL|\bm n|/a$, while the core--remainder interface contains
\begin{equation*}
O_{d,L}\left(\sum_{i=2}^d\prod_{\substack{2\leq j\leq d\\j\ne i}}n_j\right)=o(|\bm n|)
\end{equation*}
bonds as $\min_i n_i\to\infty$. The correction of the cut system is the sum of the corrections on the $M_{\bm n}$ blocks and on the remainder. The one-bond comparison and \cref{eq:pointwise-comparison} therefore give
\begin{equation}\label{eq:tiling-bound}
\left|\E\Delta F_{L,\bm n}(r)-M_{\bm n}\E\Delta F_{L,\bm a}(r)\right|
\leq2\mu_{\rm int}\left(\frac{kL}{a}|\bm n|+o(|\bm n|)\right)
+r\mu_{\rm b}|R_{\bm n}|.
\end{equation}
For fixed $a$, $M_{\bm n}a^k/|\bm n|\to1$ and $|R_{\bm n}|/|\bm n|\to0$. Dividing \cref{eq:tiling-bound} by $|\bm n|$ gives
\begin{align*}
\limsup_{\min_i n_i\to\infty}
\left|\frac{\E\Delta F_{L,\bm n}(r)}{|\bm n|}
-\frac{\E\Delta F_{L,\bm a}(r)}{a^k}\right|
&\leq\frac{2kL\mu_{\rm int}}{a},
\\
\limsup_{\min_i n_i\to\infty}
\frac{\E\Delta F_{L,\bm n}(r)}{|\bm n|}
-\liminf_{\min_i n_i\to\infty}
\frac{\E\Delta F_{L,\bm n}(r)}{|\bm n|}
&\leq\frac{4kL\mu_{\rm int}}{a}.
\end{align*}
Letting $a\to\infty$ proves existence, uniformly in $r$ and also for ground-state energies. Finally, \cref{eq:pointwise-comparison} implies
\begin{align*}
-r\mu_{\rm b}
&\leq\frac1{|\bm n|}\E\Delta F_{L,\bm n}(r)
\leq0,
\\
-\mu_{\rm b}
&\leq s_L(\beta)
\leq0.
\end{align*}
\end{proof}

\begin{proposition}[Tangential sample convergence]\label{prop:tangential-sample}
In addition to the assumptions of \cref{thm:fixed-pressure}, suppose the product law has the concentration property
\begin{equation}\label{eq:product-concentration}
\Pp(|f-\E f|\geq t)
\leq2\exp\left(-\frac{t^2}{\kappa\sum_a c_a^2}\right)
\end{equation}
for every coordinatewise Lipschitz function $f$ with coordinate constants $(c_a)$, for some $\kappa<\infty$ \cite{Talagrand1996,Ledoux2001}. On a common product environment and along tangential cubes $n_2=\cdots=n_d=n$,
\begin{equation}\label{eq:tangential-sample}
\frac{\Delta F_{L,\bm n}(r)}{n^{d-1}}\longrightarrow p_L(r)
\end{equation}
almost surely and in $L^2$, for each fixed $r$ and finite $\beta$. The same conclusion holds for $\Delta E$ at $\beta=\infty$.
\end{proposition}

\begin{proof}
For fixed $L$, there is $C_{d,L}<\infty$ such that the squared coordinate constants satisfy
\begin{align*}
\sum_a c_a^2
\leq4|E_{L,\bm n}|+r^2|\bm n|
\leq C_{d,L}n^{d-1}.
\end{align*}
With $X_n=\Delta F_{L,\bm n}(r)/n^{d-1}$, \cref{eq:product-concentration} gives
\begin{align*}
\Pp(|X_n-\E X_n|\geq\varepsilon)
&\leq2\exp\left(-\frac{\varepsilon^2n^{2(d-1)}}
{\kappa C_{d,L}n^{d-1}}\right)
\\
&=2e^{-c_{d,L,\varepsilon}n^{d-1}},
&\sum_{n\geq1}\Pp(|X_n-\E X_n|\geq\varepsilon)&<\infty.
\end{align*}
Borel--Cantelli and \cref{thm:fixed-pressure} give the almost-sure limit. Integrating the same tail yields
\begin{align*}
\E|X_n-\E X_n|^2
=\int_0^\infty2t\,
\Pp(|X_n-\E X_n|\geq t)\d t
\leq\frac{2\kappa C_{d,L}}{n^{d-1}},
\end{align*}
which proves $L^2$ convergence. Ground-state corrections have the same coordinate constants.
\end{proof}

\section{Derivative representation from concavity}

The next theorem is finite-temperature only. It is the step that removes a separate strip-interpolation hypothesis.

\begin{theorem}[Derivative representation]\label{thm:automatic-derivative}
Under the assumptions of \cref{thm:fixed-pressure}, fix $\beta<\infty$. The function $p_L:[0,1]\to\mathbb R$ exists at every $r$, is concave, and is $\mu_{\rm b}$-Lipschitz. At every differentiability point of $p_L$,
\begin{equation}\label{eq:derivative-limit}
g_L(r):=-p_L'(r)
=\lim_{\min_i n_i\to\infty}\frac{1}{|\bm n|}
\E\sum_{z\in T_{\bm n}}K_z
\left\langle\sigma_{(1,z)}\right\rangle_{L,\bm n,r}.
\end{equation}
Consequently $g_L$ is defined almost everywhere, $|g_L(r)|\leq\mu_{\rm b}$, and
\begin{equation}\label{eq:integrated-response}
s_L(\beta)=-\int_0^1g_L(r)\d r.
\end{equation}
This theorem does not assert a zero-temperature response formula.
\end{theorem}

\begin{proof}
Set
\begin{equation*}
p_{L,\bm n}(r)=\frac{1}{|\bm n|}\E\bigl[F_{L,\bm n}(r)-F_{L,\bm n}(0)\bigr].
\end{equation*}
The proof of \cref{thm:fixed-pressure} is valid for each $r$, so these functions converge pointwise to $p_L$. Since $F_{L,\bm n}(r)$ is the negative of a convex log-partition function, $p_{L,\bm n}$ is concave. Moreover,
\begin{equation}\label{eq:uniform-Lipschitz}
|p_{L,\bm n}(r)-p_{L,\bm n}(r')|
\leq|r-r'|\frac{1}{|\bm n|}\sum_{z\in T_{\bm n}}\E|K_z|
=\mu_{\rm b}|r-r'|.
\end{equation}
Pointwise convergence preserves concavity, and \cref{eq:uniform-Lipschitz} passes to the limit. Thus $p_L$ is absolutely continuous and differentiable almost everywhere.

At finite volume, dominated differentiation gives
\begin{equation}\label{eq:finite-derivative}
p_{L,\bm n}'(r)=-\frac{1}{|\bm n|}\E\sum_{z\in T_{\bm n}}K_z
\left\langle\sigma_{(1,z)}\right\rangle_{L,\bm n,r}.
\end{equation}
Fix a differentiability point $r\in(0,1)$ of $p_L$. For $h>0$ with $r\pm h\in[0,1]$, concavity gives
\begin{equation}\label{eq:secant-squeeze}
\frac{p_{L,\bm n}(r+h)-p_{L,\bm n}(r)}{h}
\leq p_{L,\bm n}'(r)
\leq\frac{p_{L,\bm n}(r)-p_{L,\bm n}(r-h)}{h}.
\end{equation}
Passing first to the tangential limit in \cref{eq:secant-squeeze} gives
\begin{align*}
\frac{p_L(r+h)-p_L(r)}{h}
&\leq\liminf_{\min_i n_i\to\infty}p_{L,\bm n}'(r)
\\
&\leq\limsup_{\min_i n_i\to\infty}p_{L,\bm n}'(r)
\leq\frac{p_L(r)-p_L(r-h)}{h}.
\end{align*}
Letting $h\downarrow0$ at a differentiability point gives
\begin{align*}
\lim_{\min_i n_i\to\infty}p_{L,\bm n}'(r)
&=p_L'(r),
\\
p_L(1)-p_L(0)
&=\int_0^1p_L'(r)\d r
=-\int_0^1g_L(r)\d r.
\end{align*}
Since $p_L(0)=0$, this is \cref{eq:integrated-response}.
\end{proof}

We set $g_L(r)=0$ on the null set where the derivative is undefined and use one-sided derivatives only when an endpoint value is mentioned. This convention does not change \cref{eq:integrated-response}.

\begin{corollary}[Gaussian squared-response formula]\label{cor:Gaussian-response}
Suppose in addition that the bottom couplings are independent centered Gaussians of variance $v$. For almost every $r\in(0,1]$, the rectangular limit
\begin{equation}\label{eq:u-L}
u_L(r)=\lim_{\min_i n_i\to\infty}\frac1{|\bm n|}
\sum_{z\in T_{\bm n}}\E\left[1-
\langle\sigma_{(1,z)}\rangle_{L,\bm n,r}^2\right]
\end{equation}
exists, satisfies $0\leq u_L(r)\leq1$, and
\begin{equation}\label{eq:Gaussian-strip-pressure}
s_L(\beta)=-\beta v\int_0^1 r u_L(r)\d r.
\end{equation}
No assertion is made about the limit $L\to\infty$ without a normal-direction mixing estimate.
\end{corollary}

\begin{proof}
At every differentiability point of $p_L$ and every $r>0$, \cref{eq:gaussian-linear-response,eq:derivative-limit} give
\begin{align*}
g_L(r)
&=\lim_{\min_i n_i\to\infty}\frac1{|\bm n|}
\E\sum_{z\in T_{\bm n}}K_z
\langle\sigma_{(1,z)}\rangle_{L,\bm n,r}
\\
&=\beta vr\lim_{\min_i n_i\to\infty}\frac1{|\bm n|}
\sum_{z\in T_{\bm n}}
\E\left[1-\langle\sigma_{(1,z)}\rangle_{L,\bm n,r}^2\right]
\\
&=\beta vr\,u_L(r).
\end{align*}
Thus $0\leq u_L(r)\leq1$ and
\begin{align*}
s_L(\beta)
=-\int_0^1g_L(r)\d r
=-\beta v\int_0^1r u_L(r)\d r.
\end{align*}
The null set where $u_L$ is undefined may be filled arbitrarily in $[0,1]$.
\end{proof}

\subsection{A Gaussian normal-localization criterion}

Let $\mathcal S_L=\{1,\ldots,L\}\times\Z^{d-1}$ be the infinite strip with the bottom Gaussian fields and no interactions across its top. For a jointly measurable, tangentially covariant family of strip states $\mu_{L,r}^\omega$ and half-space states $\mu_{\Hh,r}^\omega$, write
\begin{equation*}
m_L(r,\omega)=\mu_{L,r}^\omega(\sigma_{x_0}),
\qquad
m_{\Hh}(r,\omega)=\mu_{\Hh,r}^\omega(\sigma_{x_0}).
\end{equation*}
Call the family response-compatible if, for almost every $r\in(0,1]$,
\begin{equation}\label{eq:response-compatibility}
u_L(r)=\E[1-m_L(r,\omega)^2],
\end{equation}
where $u_L$ is the tangential limit from \cref{cor:Gaussian-response}. This condition records the required passage from finite cylinders to the infinite strip; it is not included in the finite-volume integration-by-parts identity.

We next make precise the boundary-kernel convention used for free faces.  Let
$G=(\mathcal V,\mathcal E)$ be a locally finite graph, let $V\subset\mathcal V$
be finite, and write
\begin{equation*}
\partial_E^{\rm or}V
=\{a=(x,y):x\in V,\ y\notin V,\ \{x,y\}\in\mathcal E\}.
\end{equation*}
An oriented-edge assignment is a vector
$\xi\in\{-1,1\}^{\partial_E^{\rm or}V}$; assignments on two edges with the
same exterior endpoint are not required to coincide.  An ordinary exterior
configuration $\eta$ induces the consistent assignment
$\xi_{(x,y)}=\eta_y$.  A probability kernel on the finite assignment space is
called an admissible boundary kernel.  Two such kernels are $D$-compatible if
they have a coupling supported on pairs of assignments that agree outside
$D\subset\partial_E^{\rm or}V$.

\begin{lemma}[Auxiliary edge spins and boundary kernels]
\label{lem:auxiliary-boundary-kernels}
Let $Q\subset\partial_E^{\rm or}V$ be a set of interactions to be deleted and
write
\begin{equation*}
H_V^\xi(\sigma)=H_V^0(\sigma)
-\sum_{a=(x,y)\in Q}J_a\sigma_x\xi_a,
\qquad
Z_V^\xi=\sum_\sigma e^{-\beta H_V^\xi(\sigma)},
\end{equation*}
where $H_V^0$ contains all terms not indexed by $Q$.  Give the auxiliary
spins $(\xi_a)_{a\in Q}$ the independent uniform a priori law and then apply
the Gibbs weight.  The $\sigma$-marginal is exactly the Gibbs measure with the
$Q$-interactions deleted, and it has the disintegration
\begin{equation}\label{eq:free-edge-mixture}
\gamma_V^{\rm free}(f)
=\sum_{\xi\in\{-1,1\}^Q}\kappa_{V,Q}^{\rm free}(\xi)\,
\gamma_V^\xi(f),
\qquad
\kappa_{V,Q}^{\rm free}(\xi)
=\frac{2^{-|Q|}Z_V^\xi}
{\sum_\zeta 2^{-|Q|}Z_V^\zeta}.
\end{equation}
The kernel $\kappa_{V,Q}^{\rm free}$ is jointly measurable in every parameter
on which the finite-volume Hamiltonian is measurable.  Moreover, if a
pointwise comparison estimate holds for every pair of edge assignments whose
disagreement set is contained in $D$, then the same estimate holds after
integration against any $D$-compatible pair of admissible boundary kernels.
\end{lemma}

\begin{proof}
For fixed $\sigma$, summing one auxiliary spin gives
\begin{equation*}
\frac12\sum_{s=\pm1}e^{\beta J_a\sigma_xs}
=\cosh(\beta J_a),
\end{equation*}
which is independent of $\sigma_x$.  Consequently,
\begin{equation*}
2^{-|Q|}\sum_\xi e^{-\beta H_V^\xi(\sigma)}
=e^{-\beta H_V^0(\sigma)}\prod_{a\in Q}\cosh(\beta J_a).
\end{equation*}
The product cancels on normalization, proving the assertion about the
$\sigma$-marginal.  Disintegration of the same joint law with respect to
$\xi$ gives \cref{eq:free-edge-mixture}; in particular, the mixing law is the
partition-function-weighted kernel $\kappa_{V,Q}^{\rm free}$, not the uniform
a priori law.  Joint measurability follows from the finite sums defining
$Z_V^\xi$.  Finally, if $\pi$ is a $D$-compatible coupling of $\kappa$ and
$\kappa'$, then
\begin{align*}
\left|\int\gamma_V^\xi(f)\,\kappa(\d\xi)
-\int\widetilde\gamma_V^{\xi'}(f)\,\kappa'(\d\xi')\right|
&\leq\int\left|\gamma_V^\xi(f)
-\widetilde\gamma_V^{\xi'}(f)\right|\pi(\d\xi,\d\xi'),
\end{align*}
and the asserted integrated bound follows from the pointwise one.
\end{proof}

\begin{proposition}[Gaussian localization]\label{prop:Gaussian-localization}
Suppose that a response-compatible family satisfies, for a deterministic $\rho(L)\downarrow0$,
\begin{equation}\label{eq:Gaussian-localization}
\sup_{r\in[0,1]}\E|m_L(r,\omega)-m_{\Hh}(r,\omega)|
\leq\rho(L).
\end{equation}
Then
\begin{equation}\label{eq:Gaussian-normal-limit}
s_L(\beta)\longrightarrow
\tau_{\rm G}^{\Hh}:=-\beta v\int_0^1r\,
\E[1-m_{\Hh}(r,\omega)^2]\d r,
\end{equation}
and quantitatively
\begin{equation}\label{eq:Gaussian-normal-rate}
|s_L(\beta)-\tau_{\rm G}^{\Hh}|
\leq\beta v\rho(L).
\end{equation}
The conclusion is independent of a selected half-space state whenever the squared response on the right of \cref{eq:Gaussian-normal-limit} is selection-independent.
\end{proposition}

\begin{proof}
For almost every $r$,
\begin{align*}
\left|u_L(r)-\E[1-m_{\Hh}(r,\omega)^2]\right|
&=\left|\E[m_{\Hh}(r,\omega)^2-m_L(r,\omega)^2]\right|
\\
&\leq\E\left[|m_{\Hh}-m_L|\,|m_{\Hh}+m_L|\right]
\\
&\leq2\E|m_{\Hh}-m_L|
\leq2\rho(L).
\end{align*}
Hence
\begin{align*}
|s_L(\beta)-\tau_{\rm G}^{\Hh}|
&\leq\beta v\int_0^1r
\left|u_L(r)-\E[1-m_{\Hh}(r,\omega)^2]\right|\d r
\\
&\leq2\beta v\rho(L)\int_0^1r\d r
=\beta v\rho(L).
\end{align*}
\end{proof}

For completeness, the strip specification used below is the specification on
the induced graph $\mathcal S_L$.  If $V\subset\mathcal S_L$ is finite and
$\eta\in\{-1,1\}^{\mathcal S_L\setminus V}$, set
\begin{align}
H_{V,r}^{\omega,L,\eta}(\sigma)
={}&-\sum_{\substack{\{x,y\}\in E(\mathcal S_L)\\x,y\in V}}
J_{xy}\sigma_x\sigma_y
-\sum_{\substack{\{x,y\}\in E(\mathcal S_L)\\x\in V,\ y\notin V}}
J_{xy}\sigma_x\eta_y
\notag\\
&-r\sum_{\substack{(1,z)\in V\\z\in\Z^{d-1}}}
K_z\sigma_{(1,z)},
\label{eq:strip-specification}
\end{align}
and denote its Gibbs kernel by $\gamma_{V,r}^{\omega,L}$.  In particular,
edges joining levels $L$ and $L+1$ do not occur in
\cref{eq:strip-specification}.

\begin{lemma}[Strip transfer, uniqueness, and response compatibility]
\label{lem:strip-transfer}
Assume uniform all-site one-spin mixing in the sense of \cref{def:USM}, with
common deterministic rate $m>0$ and tangentially stationary prefactors
$(A_x)$ satisfying $\E A_{x_0}<\infty$.  Then, on one common full-probability
event and for every $L\geq1$ and $r\in[0,1]$, the strip specification
$\gamma_{V,r}^{\omega,L}$ satisfies the same one-site estimate and has a
unique DLR state $\mu_{L,r}^\omega$.  The half-space state
$\mu_{\Hh,r}^\omega$ is also unique.  Both families are jointly Borel
measurable in $(\omega,r)$ and tangentially covariant, and deterministic
finite-cylinder exhaustions converge along the full sequence.

If the bottom couplings are independent centered Gaussians of variance $v$,
the strip family is response-compatible in the sense of
\cref{eq:response-compatibility}.  Moreover, with
$m_L(r,\omega)=\mu_{L,r}^\omega(\sigma_{x_0})$ and
$m_{\Hh}(r,\omega)=\mu_{\Hh,r}^\omega(\sigma_{x_0})$,
\begin{equation}\label{eq:strip-halfspace-transfer}
\sup_{r\in[0,1]}
\E|m_L(r,\omega)-m_{\Hh}(r,\omega)|
\leq c_{d-1}(m)\E A_{x_0}e^{-mL}.
\end{equation}
\end{lemma}

\begin{proof}
For finite $V\subset\mathcal S_L$, start from the half-space kernel on $V$
and declare every crossing edge from level $L$ to level $L+1$ to be free.
By \cref{lem:auxiliary-boundary-kernels}, integration over the corresponding
auxiliary edge spins deletes exactly those interactions and gives
$\gamma_{V,r}^{\omega,L}$.  The common-deletion clause in
\cref{def:USM} therefore transfers the one-site estimate, with the same rate
and prefactors, to the strip.  The sequential coupling of
\cref{lem:one-site-to-local} transfers local-observable mixing as well.
Exact DLR disintegration over increasing cylinders, followed first by the
tangential and then by the normal cutoff as in \cref{thm:DLR-unique}, proves
uniqueness for the strip and half-space specifications.

Let $T_N=[-N,N]^{d-1}\cap\Z^{d-1}$ and
$V_{L,N}=\{1,\ldots,L\}\times T_N$.  Equip $V_{L,N}$ with free lateral
faces, and extend its Gibbs measure outside the cylinder by any fixed spin
configuration.  Every subsequential weak limit is a strip DLR state: for a
fixed local DLR equation, the lateral face is absent once $N$ is sufficiently
large.  Compactness gives subsequential limits and uniqueness identifies all
of them, so the full sequence converges locally to $\mu_{L,r}^\omega$.
Finite-volume cylinder probabilities are jointly Borel in $(\omega,r)$;
their pointwise limits show that $(\omega,r)\mapsto\mu_{L,r}^\omega$ is
Borel.  The same exhaustion argument applies in the half-space.  Finally,
tangential translations carry either unique DLR state to the unique DLR
state in the translated environment, which proves covariance without making
a separate choice of selector.

Write $m_{L,N,z}$ and $m_{L,z}$ for the mean of $\sigma_{(1,z)}$ in the
finite cylinder and in the infinite strip.  Compare the free lateral-face
kernel from \cref{lem:auxiliary-boundary-kernels} with the exact DLR
disintegration of $\mu_{L,r}^\omega$.  Summing the exponential weights over
the lateral boundary gives a constant $C_{d,m}<\infty$, independent of $L$,
$N$, $z$, and $r$, such that
\begin{equation}\label{eq:strip-lateral-comparison}
\E|m_{L,N,z}-m_{L,z}|
\leq C_{d,m}\E A_{x_0}
e^{-m\dist(z,\partial T_N)}.
\end{equation}
Here stationarity was used to replace $\E A_{(1,z)}$ by $\E A_{x_0}$; the
sum over the finitely many strip levels is bounded by a geometric series.
Consequently,
\begin{align*}
\frac1{|T_N|}\sum_{z\in T_N}
\left|\E[1-m_{L,N,z}^2]-\E[1-m_{L,z}^2]\right|
\leq\frac2{|T_N|}\sum_{z\in T_N}
\E|m_{L,N,z}-m_{L,z}|\longrightarrow0.
\end{align*}
Tangential covariance makes
$\E[1-m_{L,z}^2]=\E[1-m_L^2]$.  At every differentiability point used in
\cref{cor:Gaussian-response}, finite-volume Gaussian integration by parts
and \cref{eq:u-L} now give
\begin{equation*}
u_L(r)=\E[1-m_L(r,\omega)^2],
\end{equation*}
which is response compatibility.

It remains to compare the strip and half-space states.  Disintegrate both
states over $V_{L,N}$, representing the missing strip-top interactions by
the free-edge kernel of \cref{lem:auxiliary-boundary-kernels}.  The two
boundary kernels can differ only on the top and lateral crossing edges, so
uniform one-site mixing gives
\begin{align*}
|m_L(r,\omega)-m_{\Hh}(r,\omega)|
&\leq A_{x_0}(\omega)
\sum_{z\in T_N}e^{-m\dist(x_0,(L+1,z))}
+\varepsilon_N(\omega),
\end{align*}
where the lateral term $\varepsilon_N(\omega)$ tends to zero by the same
geometric summation as in \cref{eq:strip-lateral-comparison}.  Letting
$N\to\infty$ and using \cref{eq:top-sum} yields, uniformly in $r$,
\begin{equation*}
|m_L(r,\omega)-m_{\Hh}(r,\omega)|
\leq c_{d-1}(m)A_{x_0}(\omega)e^{-mL}.
\end{equation*}
Taking expectations proves \cref{eq:strip-halfspace-transfer}.
\end{proof}

\begin{corollary}[Exponential Gaussian localization]\label{cor:Gaussian-exponential}
Assume uniform all-site one-spin mixing in the sense of \cref{def:USM}, with a common deterministic rate $m>0$ and stationary prefactors $(A_x)$ satisfying $\E A_{x_0}<\infty$. Then the strip and half-space DLR states are unique, their families can be chosen jointly measurable and response-compatible, and
\begin{equation}\label{eq:Gaussian-explicit-rho}
\rho(L)=c_{d-1}(m)\E A_{x_0}e^{-mL}
\end{equation}
is admissible in \cref{prop:Gaussian-localization}. Consequently,
\begin{equation*}
|s_L(\beta)-\tau_{\rm G}^{\Hh}|
\leq\beta v c_{d-1}(m)\E A_{x_0}e^{-mL}.
\end{equation*}
\end{corollary}

\begin{proof}
The existence, uniqueness, measurability, covariance, and full-sequence
exhaustion assertions follow from \cref{lem:strip-transfer}.  We record the
resulting estimates explicitly.  Let $m_{L,N,z}$ and $m_{L,z}$ denote the boundary-spin means in a width-$N$ cylinder and in the infinite strip. The lateral comparison gives a constant $C_{d,m}<\infty$ such that
\begin{align*}
\frac1{|T_N|}\sum_{z\in T_N}
\E|m_{L,N,z}-m_{L,z}|
&\leq\frac{C_{d,m}\E A_{x_0}}{|T_N|}
\sum_{z\in T_N}e^{-m\dist(z,\partial T_N)}
\longrightarrow0,
\\
\frac1{|T_N|}\sum_{z\in T_N}
\left|\E[1-m_{L,N,z}^2]-\E[1-m_{L,z}^2]\right|
&\leq\frac2{|T_N|}\sum_{z\in T_N}
\E|m_{L,N,z}-m_{L,z}|
\longrightarrow0.
\end{align*}
Tangential stationarity and finite-volume Gaussian integration by parts now give \cref{eq:response-compatibility}. Comparing the exact strip and half-space DLR disintegrations, first removing the lateral cutoff, leaves only the top plane:
\begin{align*}
|m_L(r,\omega)-m_{\Hh}(r,\omega)|
&\leq A_{x_0}(\omega)
\sum_{z\in\Z^{d-1}}e^{-m\dist(x_0,(L+1,z))}
\\
&=c_{d-1}(m)A_{x_0}(\omega)e^{-mL},
\\
\sup_{r\in[0,1]}\E|m_L(r,\omega)-m_{\Hh}(r,\omega)|
&\leq c_{d-1}(m)\E A_{x_0}e^{-mL}.
\end{align*}
This is \cref{eq:Gaussian-explicit-rho}.
\end{proof}

\section{Half-space response and state selection}

Writing $\mu_{\bm L,r,y}$ and $\nu_{\bm L,r,y}$ for the corresponding
one-site conditional kernels, let
\begin{equation*}
\Hh_d=\{(x_1,z)\in\Z\times\Z^{d-1}:x_1\geq1\},
\qquad x_0=(1,0).
\end{equation*}
The internal edge couplings $(J_e)_{e\in E(\Hh_d)}$ are i.i.d., and the bottom fields $(K_z)$ are i.i.d. and independent of them. For finite $V\subset\Hh_d$, a non-bottom exterior configuration $\eta$, and $r\in[0,1]$, the Hamiltonian is
\begin{align}
H_{V,r}^{\omega,\eta}(\sigma)
={}&-\sum_{\substack{\{x,y\}\in E(\Hh_d)\\x,y\in V}}J_{xy}\sigma_x\sigma_y
-\sum_{\substack{\{x,y\}\in E(\Hh_d)\\x\in V,\ y\notin V}}J_{xy}\sigma_x\eta_y
\notag\\
&-r\sum_{\substack{(1,z)\in V\\z\in\Z^{d-1}}}K_z\sigma_{(1,z)}.
\label{eq:half-H}
\end{align}
Write $\gamma_{V,r}^{\omega}(\cdot\mid\eta)$ for its Gibbs specification and $\mathcal G_r(\omega)$ for the nonempty set of DLR states \cite{Georgii2011,Ruelle1969}. A selector is disorder-measurable; a family indexed by $r$ is required to be jointly measurable whenever it is integrated in $r$.

\begin{proposition}[Jointly measurable DLR selection]\label{prop:measurable-selection}
Let
\begin{equation*}
\Omega=\mathbb R^{E(\Hh_d)}\times\mathbb R^{\Z^{d-1}},
\qquad X=\{-1,1\}^{\Hh_d},
\end{equation*}
with their product Borel structures, and let $\mathcal P(X)$ carry the topology of weak convergence. At every finite inverse temperature there is a Borel map
\begin{equation}\label{eq:measurable-selector}
(\omega,r)\longmapsto\mu_r^\omega\in\mathcal P(X),
\qquad (\omega,r)\in\Omega\times[0,1],
\end{equation}
such that $\mu_r^\omega\in\mathcal G_r(\omega)$ for every $(\omega,r)$.
\end{proposition}

\begin{proof}
Both $\Omega\times[0,1]$ and $\mathcal P(X)$ are standard Borel spaces, and $\mathcal P(X)$ is compact metrizable. Let $\mathcal F_0$ be the countable algebra of cylinder functions with rational values, and let $\mathcal V_0$ be the countable collection of finite subsets of $\Hh_d$. The DLR graph is the set of triples $(\omega,r,\mu)$ satisfying
\begin{equation}\label{eq:countable-DLR-graph}
\int h f\,\d\mu
=\int h(\eta)\,\gamma_{V,r}^{\omega}(f\mid\eta)\,\mu(\d\eta)
\end{equation}
for every $V\in\mathcal V_0$, every $f\in\mathcal F_0$ supported in $V$, and every $h\in\mathcal F_0$ supported in $V^c$. Each equality defines a Borel set: the specification depends continuously on $r$ and on the finitely many disorder coordinates meeting $V$, and integration of a bounded cylinder function is continuous in $\mu$. The countable intersection is therefore a Borel graph. Its section at $(\omega,r)$ is precisely $\mathcal G_r(\omega)$ by the monotone-class theorem; the section is nonempty by compactness of finite-volume Gibbs measures and is closed, hence compact. The Arsenin--Kunugui uniformization theorem for Borel sets with sigma-compact sections \cite[Theorem~18.18]{Kechris1995} applies directly to this graph and gives the Borel selector \cref{eq:measurable-selector}. Equivalently, its projection through every open subset of $\mathcal P(X)$ is Borel, which verifies the weak measurability required by the Kuratowski--Ryll-Nardzewski formulation. No exceptional disorder set is needed for existence.
\end{proof}

\begin{definition}[Boundary-response mixing]\label{def:BRM}
At a finite inverse temperature $\beta$, BRM holds if there are a deterministic $m_\beta>0$, a common full-probability event $\Omega_0$, and a measurable $A_\beta:\Omega_0\to[0,\infty)$ such that
\begin{equation}\label{eq:BRM-moment}
\E[|K_0|A_\beta]<\infty
\end{equation}
and, for every $\omega\in\Omega_0$, $r\in[0,1]$, finite $V\ni x_0$, and non-bottom exterior configurations $\eta,\eta'$ differing only on $D\subset\partial^{\rm ext}V$,
\begin{equation}\label{eq:BRM}
\left|\gamma_{V,r}^{\omega}(\sigma_{x_0}\mid\eta)
-\gamma_{V,r}^{\omega}(\sigma_{x_0}\mid\eta')\right|
\leq A_\beta(\omega)\sum_{y\in D}e^{-m_\beta\dist(x_0,y)}.
\end{equation}
The estimate is pointwise in the boundary configurations. It therefore remains valid after integrating against arbitrary, possibly disorder-dependent, boundary kernels.
\end{definition}

For $j\geq0$ define
\begin{equation}\label{eq:c-ell}
c_j(m)=\left(\frac{1+e^{-m}}{1-e^{-m}}\right)^j,
\qquad
\ell_d(m)=\frac{2(d-1)c_{d-2}(m)}{1-e^{-m}}.
\end{equation}
The first constant is the exact weighted cardinality of a tangential hyperplane:
\begin{equation}\label{eq:top-sum}
\sum_{z\in\Z^{d-1}}e^{-m\dist(x_0,(L+1,z))}
=c_{d-1}(m)e^{-mL}.
\end{equation}

\begin{lemma}[Exact DLR comparison]\label{lem:exact-DLR}
Assume BRM and set
$C_{M,N}=\{1,\ldots,M\}\times([-N,N]^{d-1}\cap\Z^{d-1})$. For every $\omega\in\Omega_0$, $r\in[0,1]$, and $\mu,\widetilde\mu\in\mathcal G_r(\omega)$,
\begin{align}
|\mu(\sigma_{x_0})-\widetilde\mu(\sigma_{x_0})|
\leq A_\beta(\omega)\bigl(&c_{d-1}(m_\beta)e^{-m_\beta M}
\notag\\
&+\ell_d(m_\beta)e^{-m_\beta(N+1)}\bigr).
\label{eq:DLR-comparison}
\end{align}
In particular, the two boundary-spin expectations agree. The tangential cutoff $N$ is sent to infinity before the normal depth $M$.
\end{lemma}

\begin{proof}
The DLR equations give the exact identities
\begin{equation}\label{eq:DLR-disintegration}
\mu(\sigma_{x_0})
=\int\gamma_{C_{M,N},r}^{\omega}(\sigma_{x_0}\mid\eta)\,\mu(\d\eta),
\qquad
\widetilde\mu(\sigma_{x_0})
=\int\gamma_{C_{M,N},r}^{\omega}(\sigma_{x_0}\mid\eta')\,\widetilde\mu(\d\eta').
\end{equation}
With $\pi=\mu\otimes\widetilde\mu$ on the two exterior configurations, \cref{eq:BRM} gives
\begin{align*}
|\mu(\sigma_{x_0})-\widetilde\mu(\sigma_{x_0})|
&\leq\int\left|\gamma_{C_{M,N},r}^{\omega}(\sigma_{x_0}\mid\eta)
-\gamma_{C_{M,N},r}^{\omega}(\sigma_{x_0}\mid\eta')\right|\pi(\d\eta,\d\eta')
\\
&\leq A_\beta(\omega)
\sum_{y\in D_{M,N}^{\rm top}\cup D_{M,N}^{\rm lat}}
e^{-m_\beta\dist(x_0,y)}.
\end{align*}
The two boundary sums satisfy
\begin{align*}
\sum_{y\in D_{M,N}^{\rm top}}e^{-m\dist(x_0,y)}
&\leq e^{-mM}\sum_{z\in\Z^{d-1}}e^{-m|z|_1}
=c_{d-1}(m)e^{-mM},
\\
\sum_{y\in D_{M,N}^{\rm lat}}e^{-m\dist(x_0,y)}
&\leq2(d-1)e^{-m(N+1)}
\left(\sum_{j\geq0}e^{-mj}\right)
\left(\sum_{u\in\Z^{d-2}}e^{-m|u|_1}\right)
\\
&=\ell_d(m)e^{-m(N+1)}.
\end{align*}
This proves \cref{eq:DLR-comparison}. In the prescribed order,
\begin{align*}
\lim_{M\to\infty}\lim_{N\to\infty}
|\mu(\sigma_{x_0})-\widetilde\mu(\sigma_{x_0})|
\leq\lim_{M\to\infty}
A_\beta(\omega)c_{d-1}(m_\beta)e^{-m_\beta M}
=0.
\end{align*}
\end{proof}

For a measurable selection $\mu_r^\omega\in\mathcal G_r(\omega)$ define
\begin{equation}\label{eq:half-response}
b_\mu(r)=\E\left[K_0\mu_r^\omega(\sigma_{x_0})\right].
\end{equation}
For a jointly measurable family define
\begin{equation}\label{eq:half-pressure}
\tau^{\Hh_d}(\mu)=-\int_0^1b_\mu(r)\d r.
\end{equation}
Both expressions are integrable because $|b_\mu(r)|\leq\E|K_0|$.

\begin{theorem}[Selection-independent boundary response]\label{thm:selection-independent}
Assume BRM. For every fixed $r$, any two disorder-measurable DLR selections have the same $b_\mu(r)$. Any two jointly measurable selection families have the same $\tau^{\Hh_d}$. This is selection independence of one boundary response; it is not a claim of DLR uniqueness.
\end{theorem}

\begin{proof}
For two selections $\mu$ and $\widetilde\mu$, \cref{eq:DLR-comparison} gives
\begin{align*}
|b_\mu(r)-b_{\widetilde\mu}(r)|
&\leq\E\left[|K_0|\,
|\mu_r^\omega(\sigma_{x_0})-\widetilde\mu_r^\omega(\sigma_{x_0})|\right]
\\
&\leq\E[|K_0|A_\beta]
\left(c_{d-1}(m_\beta)e^{-m_\beta M}
+\ell_d(m_\beta)e^{-m_\beta(N+1)}\right).
\end{align*}
The moment condition permits first $N\to\infty$ and then $M\to\infty$, so $b_\mu(r)=b_{\widetilde\mu}(r)$. Consequently,
\begin{align*}
|\tau^{\Hh_d}(\mu)-\tau^{\Hh_d}(\widetilde\mu)|
\leq\int_0^1|b_\mu(r)-b_{\widetilde\mu}(r)|\d r
=0.
\end{align*}
\end{proof}

\begin{theorem}[Finite-depth approximation]\label{thm:depth-approximation}
Assume the finite-first-moment product model of \cref{thm:fixed-pressure}, let $\beta<\infty$, and assume BRM. The jointly measurable selection supplied by \cref{prop:measurable-selection} has a half-space pressure independent of the chosen jointly measurable family, and
\begin{equation}\label{eq:depth-rate}
\left|s_L(\beta)-\tau^{\Hh_d}\right|
\leq c_{d-1}(m_\beta)\E[|K_0|A_\beta]e^{-m_\beta L}.
\end{equation}
The tangential limit is taken first for every fixed $L$; \cref{eq:depth-rate} then permits the normal-depth limit. No all-face cube theorem follows from this one-face hypothesis alone.
\end{theorem}

\begin{proof}
We use centered tangential boxes $T_N=[-N,N]^{d-1}\cap\Z^{d-1}$, which have the same limiting pressure by stationarity and \cref{thm:fixed-pressure}. For every deleted top or lateral edge, adjoin an auxiliary edge spin with independent uniform a priori law. Summing that spin contributes the spin-independent factor $\cosh(\beta J_e)$. By \cref{lem:auxiliary-boundary-kernels}, the resulting spin marginal represents the free faces exactly and its disintegration is a mixture of fixed-edge boundary responses against the induced partition-function-weighted kernel. The spin-independent factors cancel from normalized expectations. On the rectangular cylinders used here, every deleted crossing edge has a distinct exterior endpoint, so these edge-assignment kernels are ordinary exterior-configuration kernels covered by BRM. Hence the finite-slab response can be compared with the exact DLR disintegration of a half-space state.

An arbitrary selector need not be tangentially covariant. We therefore make the translation argument explicit. Let $\tau_z(x_1,u)=(x_1,u+z)$ and define
\begin{equation*}
J_{\{x,y\}}(\vartheta_z\omega)=J_{\{\tau_zx,\tau_zy\}}(\omega),
\qquad
K_u(\vartheta_z\omega)=K_{u+z}(\omega).
\end{equation*}
For a spin configuration let $(S_z\sigma)_{\tau_zx}=\sigma_x$. If $\mu^{\vartheta_z\omega}$ is the selected DLR state for the shifted environment, then
\begin{equation}\label{eq:translated-selector}
\mu^{[z],\omega}:=(S_z)_*\mu^{\vartheta_z\omega}
\end{equation}
is a DLR state for $\omega$ and
$\mu^{[z],\omega}(\sigma_{(1,z)})=\mu^{\vartheta_z\omega}(\sigma_{x_0})$.
Replace $\Omega_0$ by the full-probability event
$\Omega_*=\bigcap_{z\in\Z^{d-1}}\vartheta_z^{-1}\Omega_0$ and set
$A_z(\omega)=A_\beta(\vartheta_z\omega)$. Stationarity gives
\begin{equation}\label{eq:stationary-selector}
\E\left[K_z\mu^{[z],\omega}(\sigma_{(1,z)})\right]=b_\mu(r),
\qquad
\E[|K_z|A_z]=\E[|K_0|A_\beta].
\end{equation}
Thus selector covariance has not been assumed.

Define the averaged finite-slab response by
\begin{align*}
g_{L,N}(r)
=\frac1{|T_N|}\E\sum_{z\in T_N}K_z
\langle\sigma_{(1,z)}\rangle_{L,N,r}.
\end{align*}
Applying BRM at every $(1,z)$ to the free-face boundary mixture and the DLR disintegration of \cref{eq:translated-selector} gives a constant $C_{d,m_\beta}$ such that
\begin{align*}
|g_{L,N}(r)-b_\mu(r)|
&\leq c_{d-1}(m_\beta)\E[|K_0|A_\beta]e^{-m_\beta L}
\\
&\quad+\frac{C_{d,m_\beta}\E[|K_0|A_\beta]}{|T_N|}
\sum_{z\in T_N}e^{-m_\beta\dist(z,\partial T_N)}.
\label{eq:finite-depth-response-bound}
\end{align*}
For centered rectangular boxes,
\begin{align*}
\sum_{z\in T_N}e^{-m_\beta\dist(z,\partial T_N)}
&=O_{d,m_\beta}(N^{d-2}),
\\
\frac1{|T_N|}\sum_{z\in T_N}e^{-m_\beta\dist(z,\partial T_N)}
&=O_{d,m_\beta}(N^{-1})
\longrightarrow0,
\end{align*}
where $|T_N|=(2N+1)^{d-1}$.\\
Sending $N\to\infty$ at fixed $L$ and using \cref{thm:automatic-derivative} gives, for almost every $r$,
\begin{equation}\label{eq:g-b-bound}
|g_L(r)-b_\mu(r)|
\leq c_{d-1}(m_\beta)\E[|K_0|A_\beta]e^{-m_\beta L}.
\end{equation}
Since both integrands are bounded by $\mu_{\rm b}$,
\begin{align*}
|s_L(\beta)-\tau^{\Hh_d}|
&=\left|\int_0^1\bigl(g_L(r)-b_\mu(r)\bigr)\d r\right|
\\
&\leq c_{d-1}(m_\beta)\E[|K_0|A_\beta]e^{-m_\beta L}.
\end{align*}
Selection independence is \cref{thm:selection-independent}.
\end{proof}

\section{Mixing hierarchy and a verifiable regime}

BRM controls only the response observable at the distinguished boundary site. We first record an abstract one-site hypothesis that is uniform enough to survive successive conditioning.

\begin{definition}[Uniform all-site one-spin mixing]\label{def:USM}
At a finite inverse temperature $\beta$, uniform all-site one-spin mixing holds if there are a deterministic $m_\beta>0$, a common full-probability event, and finite measurable prefactors $(A_x)_{x\in\Hh_d}$ such that, for every environment on that event, every finite $V\subset\Hh_d$, $x\in V$, $r\in[0,1]$, and non-bottom exterior configurations differing only on $D\subset\partial^{\rm ext}V$,
\begin{equation}\label{eq:USM}
\left|\gamma_{V,r}^{\omega}(\sigma_x\mid\eta)
-\gamma_{V,r}^{\omega}(\sigma_x\mid\eta')\right|
\leq A_x(\omega)\sum_{y\in D}
e^{-m_\beta\dist(x,y)}.
\end{equation}
The estimate is required uniformly after replacing any common subset of spins by fixed boundary data. In the oriented-edge formulation preceding \cref{lem:auxiliary-boundary-kernels}, it is also required pointwise for edge assignments and after integration against paired $D$-compatible admissible boundary kernels. These include the induced free-edge mixtures generated by independent uniform auxiliary edge spins; the induced mixing law is the partition-function-weighted kernel in \cref{eq:free-edge-mixture}, and edge spins incident to the same exterior vertex may be treated separately. In addition, if the same set $Q$ of crossing interactions is declared free in both specifications by the construction of \cref{lem:auxiliary-boundary-kernels}, the estimate is required with $D$ containing only the boundary labels on which the remaining data differ. This common-deletion clause is what transfers the hypothesis to induced subgraphs such as $\mathcal S_L$; it is not a consequence of ordinary vertex-boundary mixing alone. The source hypothesis with one integrable random prefactor $A_\beta$ is the special case $A_x\leq A_\beta$ for all $x$.
\end{definition}

To identify an entire Gibbs state, one needs a corresponding estimate on a generating family.

\begin{definition}[Local-observable mixing]\label{def:LOM}
Let $\mathcal C$ be the countable family of cylinder indicators
$f_{S,\xi}=\mathbf1\{\sigma_S=\xi\}$, where $S\subset\Hh_d$ is finite and $\xi\in\{-1,1\}^S$. Local-observable mixing holds on a common full-probability event if, for every $f\in\mathcal C$, there are $m_{\beta,f}>0$ and a finite measurable $A_{\beta,f}$ such that, whenever $V\supset S$ and two non-bottom exterior configurations differ only on $D$,
\begin{equation}\label{eq:LOM}
\left|\gamma_{V,r}^{\omega}(f\mid\eta)-
\gamma_{V,r}^{\omega}(f\mid\eta')\right|
\leq A_{\beta,f}(\omega)\osc(f)
\sum_{y\in D}e^{-m_{\beta,f}\dist(S,y)}.
\end{equation}
The estimate is uniform in $V$, $r\in[0,1]$, and the two boundary configurations.
\end{definition}

\begin{lemma}[One-site mixing lifts to local observables]\label{lem:one-site-to-local}
Uniform all-site one-spin mixing implies local-observable mixing. More precisely, for $f_{S,\xi}\in\mathcal C$ one may take
\begin{equation}\label{eq:USM-to-LOM}
m_{\beta,f}=m_\beta,
\qquad
A_{\beta,f}=\frac12\sum_{x\in S}A_x,
\end{equation}
and there is one coupling of the two $S$-marginals for which
\begin{equation}\label{eq:finite-set-exposure}
\Pp(\sigma_S\ne\sigma_S')
\leq\frac12\sum_{x\in S}\sum_{y\in D}
A_xe^{-m_\beta\dist(x,y)}.
\end{equation}
The same conclusion holds after integrating against paired, possibly disorder-dependent, admissible boundary kernels.
\end{lemma}

\begin{proof}
Order $S=\{x_1,\ldots,x_k\}$. Couple the two spins at $x_1$ maximally. Their total-variation distance is one half of the absolute difference of their means, so \cref{eq:USM} bounds the first mismatch probability by
\begin{align*}
\Pp(\sigma_{x_1}\ne\sigma_{x_1}')
\leq\frac12A_{x_1}\sum_{y\in D}
e^{-m_\beta\dist(x_1,y)}.
\end{align*}
On the branch where the first spins agree, condition on their common value and remove $x_1$ from the volume. The two conditional laws are exact Gibbs specifications on the reduced volume, with $x_1$ added as common boundary data; hence \cref{eq:USM} applies at $x_2$ with the same changed set $D$. Continue in this way. On a branch where a mismatch has already occurred, complete the coupling arbitrarily. Finite-state gluing gives one coupling of the $S$-marginals.

If $E_i$ is the event that the first mismatch occurs at the $i$th exposure, then
\begin{align*}
\Pp(\sigma_S\ne\sigma_S')
&=\sum_{i=1}^k\Pp(E_i)
\leq\frac12\sum_{i=1}^kA_{x_i}
\sum_{y\in D}e^{-m_\beta\dist(x_i,y)},
\\
|\gamma_{V,r}^\omega(f\mid\eta)-\gamma_{V,r}^\omega(f\mid\eta')|
&\leq\osc(f)\Pp(\sigma_S\ne\sigma_S')
\\
&\leq\frac12\left(\sum_{x\in S}A_x\right)\osc(f)
\sum_{y\in D}e^{-m_\beta\dist(S,y)}.
\end{align*}
Integration against a coupling of two admissible boundary kernels preserves the same bound.
\end{proof}

\begin{theorem}[DLR uniqueness]\label{thm:DLR-unique}
Assume local-observable mixing. On its common full-probability event, $\mathcal G_r(\omega)$ consists of one DLR state for every $r\in[0,1]$. This is strictly stronger than the response conclusion of \cref{thm:selection-independent}.
\end{theorem}

\begin{proof}
Fix $f_{S,\xi}\in\mathcal C$ and two DLR states $\mu,\nu$. Exact disintegration in $C_{M,N}\supset S$ and \cref{eq:LOM} give
\begin{align*}
|\mu(f_{S,\xi})-\nu(f_{S,\xi})|
\leq A_{\beta,f}(\omega)\osc(f)
&\sum_{y\in D_{M,N}^{\rm top}\cup D_{M,N}^{\rm lat}}
e^{-m_{\beta,f}\dist(S,y)},
\\
\lim_{M\to\infty}\lim_{N\to\infty}
\sum_{y\in D_{M,N}^{\rm top}\cup D_{M,N}^{\rm lat}}
&e^{-m_{\beta,f}\dist(S,y)}
=0.
\end{align*}
Thus $\mu(f)=\nu(f)$ on the common countable generating class, hence on the full configuration sigma-field. DLR existence supplies the unique state.
\end{proof}

We now verify both mixing levels in a subcritical disagreement regime by disagreement exploration \cite{vanDenBergMaes1994}. Put
\begin{equation}\label{eq:p-q}
p_\beta=\E\tanh(\beta|J|),
\qquad q_{\beta,d}=(2d-1)p_\beta.
\end{equation}
For a finite $V\subset\Hh_d$, write $\partial^{\rm ext}V$ for its non-bottom exterior vertex boundary. If $x\in V$ and $y\in\partial^{\rm ext}V$, then $\SAW_V(x,y)$ denotes the set of nearest-neighbor paths
$\gamma=(x=x_0,x_1,\ldots,x_n=y)$ with distinct vertices, $x_j\in V$ for $0\leq j<n$, and terminal vertex $y$. Thus the path runs in $V$ until its last crossing edge; all its edges carry internal couplings $J_e$.

\begin{lemma}[Single-edge influence]\label{lem:edge-influence}
Changing one neighboring spin across an edge with coupling $J$ changes the one-site conditional law at the other endpoint by at most $\tanh(\beta|J|)$ in total variation.
\end{lemma}

\begin{proof}
If $h$ is the field from all other terms, the two plus probabilities differ by
\begin{align*}
\frac12\left|\tanh(\beta(h+J))-\tanh(\beta(h-J))\right|
&=\frac{\sinh(2\beta|J|)}
{2\cosh(\beta(h+J))\cosh(\beta(h-J))}
\\
&\leq\frac{\sinh(2\beta|J|)}{2\cosh^2(\beta|J|)}
=\tanh(\beta|J|).
\end{align*}
A spin-mean difference is twice this total-variation distance.
\end{proof}

Put $c_e=\tanh(\beta|J_e|)$. Let $\partial_E^{\rm or}V$ be the oriented non-bottom crossing edges $b=(v,y)$ with $v\in V$ and $y\notin V$. To telescope changes edge by edge, temporarily allow an assignment $\xi_b\in\{-1,1\}$ on every $b\in\partial_E^{\rm or}V$, without requiring assignments on edges incident to the same exterior vertex to coincide. Write $\gamma_{V,r}^{\omega,\xi}$ for the resulting specification. For $b=(v,y)$, let $\SAW_V(x,b)$ be the paths in $\SAW_V(x,y)$ whose last edge is $b$.

\begin{lemma}[Oriented-edge disagreement coupling]\label{lem:exploration}
Let two edge assignments $\xi,\xi'$ differ on $B\subset\partial_E^{\rm or}V$. There is one coupling $Q_{\xi,\xi'}$ of
$\gamma_{V,r}^{\omega,\xi}$ and $\gamma_{V,r}^{\omega,\xi'}$ such that, simultaneously for every finite $S\subset V$,
\begin{equation}\label{eq:exploration-coupling}
Q_{\xi,\xi'}(\sigma_S\ne\sigma_S')
\leq\sum_{x\in S}\sum_{b\in B}
\sum_{\gamma\in\SAW_V(x,b)}\prod_{e\in\gamma}c_e.
\end{equation}
The estimate is uniform in the common edge assignments, the bottom fields, $r$, and all coupling signs.
\end{lemma}

\begin{proof}
We use induction on $|V|$, first for a single changed oriented edge and then for an arbitrary finite set of such edges. The assertion is trivial for $V=\varnothing$. Assume the arbitrary-edge assertion for every region with fewer than $n$ vertices, and let $|V|=n$.

First suppose that only $b=(v,y)$ changes. Put $W=V\setminus\{v\}$ and sum the common Hamiltonian over $\sigma_W$. The two marginal laws of $\sigma_v$ are one-spin laws in the same effective field, with the contribution across $b$ changed from $J_b\xi_b$ to $J_b\xi_b'$. By \cref{lem:edge-influence}, a maximal coupling makes the two spins at $v$ disagree with probability at most $c_b$. If they agree, the two conditional laws on $W$ are identical, because the only changed Boltzmann factor is constant after $\sigma_v$ is fixed, and we couple the remainder identically.

If the two spins at $v$ disagree, condition on their two values. The conditional laws on $W$ are edge-assignment specifications that can differ only on the oriented edges $(w,v)$ with $w\in W$ and $w\sim v$. Apply the induction hypothesis on $W$ to couple them. This construction gives one coupling, independent of $S$. If $v\in S$, then
\begin{align*}
Q(\sigma_S\ne\sigma_S')
&\leq Q(\sigma_v\ne\sigma_v')
\leq c_b,\qquad v\in S,
\\
Q(\sigma_S\ne\sigma_S')
&\leq c_b\sum_{x\in S}\sum_{\substack{w\in W\\w\sim v}}
\sum_{\gamma\in\SAW_W(x,(w,v))}\prod_{e\in\gamma}c_e
\\
&=\sum_{x\in S}\sum_{\gamma\in\SAW_V(x,b)}
\prod_{e\in\gamma}c_e, \qquad v\notin S.
\end{align*}
The last equality appends $b$ to a path ending at the exterior vertex $v$ of $W$. This proves the single-edge assertion for $n$ vertices, simultaneously for all $S$.

Now list the edges of $B$ as $b_1,\ldots,b_m$ and introduce assignments
$\xi=\xi^{(0)},\xi^{(1)},\ldots,\xi^{(m)}=\xi'$ by changing one oriented edge at a time. Couple each adjacent pair by the single-edge construction. The finite-state gluing lemma, obtained by disintegrating successive pair couplings over their common marginal, produces a joint law of
$(\sigma^{(0)},\ldots,\sigma^{(m)})$ with all these adjacent couplings. If $\sigma_S^{(0)}\ne\sigma_S^{(m)}$, then some adjacent pair differs on $S$, and
\begin{align*}
Q(\sigma_S^{(0)}\ne\sigma_S^{(m)})
&\leq\sum_{i=1}^m
Q(\sigma_S^{(i-1)}\ne\sigma_S^{(i)})
\\
&\leq\sum_{x\in S}\sum_{i=1}^m
\sum_{\gamma\in\SAW_V(x,b_i)}\prod_{e\in\gamma}c_e.
\end{align*}
This completes the induction. The construction is the edge-inhomogeneous signed analogue of disagreement exploration \cite{vanDenBergMaes1994}.
\end{proof}

\begin{proposition}[Pathwise disagreement estimate]\label{prop:pathwise}
For every realization, finite $V\ni x$, $r\in[0,1]$, and $D\subset\partial^{\rm ext}V$,
\begin{equation}\label{eq:pathwise}
\sup_{\substack{\eta,\eta'\\
\eta=\eta'\text{ on }\partial^{\rm ext}V\setminus D}}
\left|\gamma_{V,r}^{\omega}(\sigma_x\mid\eta)
-\gamma_{V,r}^{\omega}(\sigma_x\mid\eta')\right|
\leq2\sum_{y\in D}\sum_{\gamma\in\SAW_V(x,y)}\prod_{e\in\gamma}c_e.
\end{equation}
The estimate is taken before disorder expectation and is uniform over the boundary data. By integration it also holds for arbitrary random boundary kernels, uniformly over their disorder-dependent mixing weights. If a common set $Q$ of crossing interactions is deleted, the same estimate holds with the path sum restricted to paths that avoid $Q$, and hence with the displayed right side unchanged.
\end{proposition}

\begin{proof}
For exterior configurations $\eta,\eta'$, assign the value $\eta_y$ or $\eta_y'$ separately to every crossing edge incident to $y$. The set $B$ of changed oriented edges is precisely the collection of such edges incident to vertices in $D$ on which the two exterior spins differ. The sum over $b\in B$ in \cref{eq:exploration-coupling} is the sum over paths ending in $D$, including separate last edges when one exterior vertex is incident to several edges of a non-rectangular region. Under that coupling,
\begin{equation*}
\left|\E_Q\sigma_x-\E_Q\sigma_x'\right|
\leq2Q(\sigma_x\ne\sigma_x').
\end{equation*}
This gives \cref{eq:pathwise}. For a cylinder function $f$ supported on $S$, the same coupling gives
\begin{align}
\left|\gamma_{V,r}^{\omega}(f\mid\eta)-
\gamma_{V,r}^{\omega}(f\mid\eta')\right|
&\leq\osc(f)Q(\sigma_S\ne\sigma_S')
\notag\\
&\leq\osc(f)\sum_{x\in S}\sum_{y\in D}
\sum_{\gamma\in\SAW_V(x,y)}\prod_{e\in\gamma}c_e.
\label{eq:local-pathwise}
\end{align}
The right sides are independent of $\eta$ and $\eta'$. Integrating the pointwise inequalities against any pair of deterministic or disorder-dependent boundary kernels preserves them without changing the constants. A deleted crossing interaction has zero single-edge influence; equivalently, set $c_e=0$ on every $e\in Q$ in the exploration. The resulting restricted path sum is bounded by \cref{eq:pathwise,eq:local-pathwise}, which verifies the common-deletion clause of \cref{def:USM} in this regime.
\end{proof}

\begin{theorem}[Subcritical disagreement criterion]\label{thm:subcritical}
Assume the internal couplings are i.i.d. and $q_{\beta,d}<1$. For
$C_{M,N}=\{1,\ldots,M\}\times[-N,N]^{d-1}$, the pathwise supremum in \cref{eq:pathwise}, with boundary disagreement allowed on the top and lateral faces but not the bottom, satisfies
\begin{equation}\label{eq:annealed-path}
\E\sup_{r,\eta,\eta'}
\left|\gamma_{C_{M,N},r}^{\omega}(\sigma_{x_0}\mid\eta)
-\gamma_{C_{M,N},r}^{\omega}(\sigma_{x_0}\mid\eta')\right|
\leq\frac{4dp_\beta}{1-q_{\beta,d}}
\left(q_{\beta,d}^{M-1}+q_{\beta,d}^{N}\right).
\end{equation}
For every $m\in(0,-\log q_{\beta,d})$, define on the full half-space
\begin{equation}\label{eq:G-A}
G_{x,y}=\sum_{\gamma:x\to y\ \SAW}\prod_{e\in\gamma}c_e,
\qquad
A_x=2\sup_{y\ne x}e^{m\dist(x,y)}G_{x,y}.
\end{equation}
Then
\begin{equation}\label{eq:A-integrable}
\E A_x\leq\frac{4dp_\beta e^m}{1-q_{\beta,d}e^m}<\infty,
\end{equation}
and, on a common full-probability event for all $x$,
\begin{equation}\label{eq:quenched-mixing}
\left|\gamma_{V,r}^{\omega}(\sigma_x\mid\eta)
-\gamma_{V,r}^{\omega}(\sigma_x\mid\eta')\right|
\leq A_x(\omega)\sum_{y\in D}e^{-m\dist(x,y)}.
\end{equation}
Thus local-observable mixing holds with a prefactor no larger than $\sum_{x\in S}A_x$, and the half-space DLR state is unique. If the bottom couplings are independent of the internal field and $\E|K_0|<\infty$, then BRM holds and
\begin{equation}\label{eq:subcritical-depth}
|s_L(\beta)-\tau^{\Hh_d}|
\leq\E|K_0|\frac{4dp_\beta}{1-q_{\beta,d}}q_{\beta,d}^{L-1}.
\end{equation}
\end{theorem}

\begin{proof}
A self-avoiding path of length $n$ uses distinct edges, so independence gives
\begin{align*}
\E\prod_{e\in\gamma}c_e
=p_\beta^n,
\end{align*}
where $\#\{\gamma:|\gamma|=n, \gamma(0)=x_0\}\leq 2d(2d-1)^{n-1}$.

Top-reaching paths have length at least $M$, while lateral-reaching paths have length at least $N+1$. Including the factor $2$ from \cref{eq:pathwise},
\begin{align*}
2\sum_{n\geq M}2d(2d-1)^{n-1}p_\beta^n
&=\frac{4dp_\beta}{1-q_{\beta,d}}q_{\beta,d}^{M-1},
\\
2\sum_{n\geq N+1}2d(2d-1)^{n-1}p_\beta^n
&=\frac{4dp_\beta}{1-q_{\beta,d}}q_{\beta,d}^{N}.
\end{align*}
Their sum is \cref{eq:annealed-path}.

For the quenched estimate, $\dist(x,y)\leq|\gamma|$ and the supremum is bounded by the sum over all endpoints. Hence
\begin{align*}
\E A_x
&\leq2\sum_{n\geq1}2d(2d-1)^{n-1}p_\beta^ne^{mn}
\\
&=4dp_\beta e^m\sum_{n\geq1}
(q_{\beta,d}e^m)^{n-1}
=\frac{4dp_\beta e^m}{1-q_{\beta,d}e^m}.
\end{align*}
Countability gives a common event on which every $A_x$ is finite. On that event,
\begin{align*}
2G_{x,y}
&\leq A_xe^{-m\dist(x,y)},
\\
2\sum_{y\in D}\sum_{\gamma\in\SAW_V(x,y)}\prod_{e\in\gamma}c_e
&\leq A_x\sum_{y\in D}e^{-m\dist(x,y)},
\end{align*}
which proves \cref{eq:quenched-mixing}. Similarly, \cref{eq:local-pathwise} gives
\begin{align*}
|\gamma_{V,r}^\omega(f\mid\eta)-\gamma_{V,r}^\omega(f\mid\eta')|
\leq\frac12\left(\sum_{x\in S}A_x\right)\osc(f)
\sum_{y\in D}e^{-m\dist(S,y)}.
\end{align*}
Thus local-observable mixing and \cref{thm:DLR-unique} give uniqueness.

Independence of $K_0$ and the internal field gives
\begin{align*}
\E[|K_0|A_{x_0}]
=\E|K_0|\E A_{x_0}<\infty,
\end{align*}
so BRM follows. Retaining only top-reaching paths after the tangential cutoff is removed gives, for almost every $r$,
\begin{align*}
|g_L(r)-b_\mu(r)|
&\leq\E|K_0|\frac{4dp_\beta}{1-q_{\beta,d}}
q_{\beta,d}^{L-1},
\\
|s_L(\beta)-\tau^{\Hh_d}|
&\leq\int_0^1|g_L(r)-b_\mu(r)|\d r
\\
&\leq\E|K_0|\frac{4dp_\beta}{1-q_{\beta,d}}
q_{\beta,d}^{L-1}.
\end{align*}
\end{proof}

\begin{corollary}[Sparse signed couplings]\label{cor:sparse}
Let
\begin{equation*}
\nu_\theta=(1-\theta)\delta_0+\frac{\theta}{2}(\delta_{J_0}+\delta_{-J_0}),
\qquad J_0>0.
\end{equation*}
If $(2d-1)\theta<1$, then
$q_{\beta,d}=(2d-1)\theta\tanh(\beta J_0)<1$ for every finite $\beta$. The DLR state is unique, the response is selection-independent, and \cref{eq:subcritical-depth} holds. The admissible exponential rate remains positive as $\beta\to\infty$. This is a finite-temperature response statement and does not itself define a zero-temperature DLR response.
\end{corollary}

\section{Regular rectangles in the Dobrushin regime}

This section considers a box with all faces coupled to fixed exterior spins. Let
\begin{equation*}
\Lambda_{\bm L}=\prod_{i=1}^d\{1,\ldots,L_i\},\qquad
V_{\bm L}=\prod_{i=1}^dL_i,\qquad
\ell_{\bm L}=\min_iL_i,
\end{equation*}
let $E_{\bm L}$ be its internal edges, and let $B_{\bm L}$ be its oriented crossing edges. Thus
\begin{equation}\label{eq:rectangle-boundary-count}
b_{\bm L}:=|B_{\bm L}|=2V_{\bm L}\sum_{i=1}^dL_i^{-1}.
\end{equation}
Use one i.i.d. symmetric edge field $(J_e)$ on $\Z^d$, with $|J_e|\leq J_*$, for internal and crossing couplings; on $B_{\bm L}$ we also write $K_e:=J_e$ to distinguish the boundary role. The exterior spins are plus, and $\Delta F_{\bm L}$ denotes the simultaneous free-to-fixed correction.

\begin{lemma}[Dobrushin comparison]\label{lem:Dobrushin}
Let two finite-volume Ising specifications have a common influence matrix $C$ whose row sums are at most $\alpha<1$, with $C_{xy}=0$ for non-neighbors. If the specifications agree on the closed radius-$(R-1)$ neighborhood of the support of a local observable $f$, then
\begin{equation}\label{eq:Dobrushin-comparison}
|\mu(f)-\nu(f)|
\leq\frac{\alpha^R}{1-\alpha}\sum_x\delta_x(f),
\end{equation}
where $\delta_x(f)$ is the single-site oscillation.
\end{lemma}

\begin{proof}
Let $b_y$ be the one-site specification discrepancy. Dobrushin's comparison theorem \cite{Georgii2011} and agreement within distance $R-1$ give
\begin{align*}
|\mu(f)-\nu(f)|
&\leq\sum_x\delta_x(f)\sum_y[(I-C)^{-1}]_{xy}b_y
\\
&=\sum_x\delta_x(f)\sum_y\sum_{j\geq R}(C^j)_{xy}b_y
\\
&\leq\sum_x\delta_x(f)\sum_{j\geq R}\alpha^j
=\frac{\alpha^R}{1-\alpha}\sum_x\delta_x(f).
\end{align*}
\end{proof}

\begin{lemma}[Boundary-layer edge count]\label{lem:boundary-layer-count}
For $j\geq0$ define the inner rectangle
\begin{equation*}
\Lambda_{\bm L}^{(j)}
=\prod_{i=1}^d\{j+1,\ldots,L_i-j\},
\end{equation*}
with the convention that it is empty if one factor is empty, and let
$S_j=\Lambda_{\bm L}^{(j)}\setminus\Lambda_{\bm L}^{(j+1)}$.
Then $|S_j|\leq b_{\bm L}$.  If
$j(a)=\dist(a,\partial_{\rm in}\Lambda_{\bm L})$ for an internal edge $a$,
the number of internal edges satisfying $j(a)=j$ is at most
$2d b_{\bm L}$.
\end{lemma}

\begin{proof}
Every site of $S_j$ lies on one of the $2d$ faces of the inner rectangle, so
the union bound gives
\begin{align*}
|S_j|
&\leq2\sum_{i=1}^d\prod_{k\ne i}\max\{L_k-2j,0\}
\leq2\sum_{i=1}^d\prod_{k\ne i}L_k
=b_{\bm L}.
\end{align*}
If $j(a)=j$, at least one endpoint of $a$ belongs to $S_j$.  Each site is
incident to at most $2d$ lattice edges, hence the number of such internal
edges is at most $2d|S_j|\leq2db_{\bm L}$.
\end{proof}

\begin{theorem}[Regular-rectangle surface limit]\label{thm:Dobrushin-cube}
Let $d\geq2$ and assume
\begin{equation}\label{eq:Dobrushin-condition}
\alpha:=2d\tanh(\beta J_*)<1.
\end{equation}
There is a deterministic $\tau_{\rm Dob}(\beta)$ such that, along every sequence of rectangles with $\ell_{\bm L}\to\infty$ and for every $1\leq p<\infty$,
\begin{equation}\label{eq:cube-convergence}
\frac{\E\Delta F_{\bm L}}{b_{\bm L}}\longrightarrow\tau_{\rm Dob}(\beta),
\qquad
\frac{\Delta F_{\bm L}}{b_{\bm L}}\longrightarrow\tau_{\rm Dob}(\beta)
\quad\text{in }L^p.
\end{equation}
The second convergence is also almost sure along any sequence satisfying
$\sum_n\exp(-c b_{\bm L_n})<\infty$ for every $c>0$; in particular it is almost sure along cubes.
Let $\langle\cdot\rangle_{\Hh_d,r}$ be the unique half-space state with the bottom crossing bonds multiplied by $r$, and let $e_0$ be the bottom edge incident to $x_0$. Then
\begin{equation}\label{eq:Dob-tau}
\tau_{\rm Dob}(\beta)
=-\int_0^1\E\left[K_{e_0}\langle\sigma_{x_0}\rangle_{\Hh_d,r}\right]\d r.
\end{equation}
If $v=\Var(K_{e_0})$, then $-\beta v/2\leq\tau_{\rm Dob}(\beta)\leq0$. Moreover,
\begin{equation}\label{eq:Dob-tail}
\Pp(|\Delta F_{\bm L}-\E\Delta F_{\bm L}|\geq u)
\leq2\exp\left(-\frac{2u^2}{C_*b_{\bm L}}\right),
\end{equation}
where
\begin{equation}\label{eq:C-star}
C_*=(2J_*)^2+
\frac{32d(2J_*)^2}{(1-\alpha)^2(1-\alpha^2)}.
\end{equation}
This theorem uses all-site Dobrushin comparison and does not extend to unbounded Gaussian couplings. It is independent of the one-face BRM theorem.
All conclusions remain valid if the plus exterior spins are replaced by any deterministic, disorder-independent exterior pattern.
\end{theorem}

\begin{proof}
For the Ising specification one may take
$C_{xy}=\tanh(\beta J_*)\mathbf1_{\{x\sim y\}}$. Boundary couplings to prescribed exterior spins are one-site fields and do not enter $C$. The comparison estimate is uniform in the boundary amplitude, disorder realization, and exterior spins. It constructs the unique half-space state by exhaustion and proves measurability as a pointwise limit of finite-volume expectations.

Differentiate the simultaneous linear interpolation on all crossing edges:
\begin{equation}\label{eq:cube-linear}
\E\Delta F_{\bm L}
=-\int_0^1\sum_{e\in B_{\bm L}}\E[K_e\langle q_e\rangle_{\bm L,r}]\d r.
\end{equation}

Fix an oriented face. A lattice isometry maps its inward normal to the first coordinate direction and its incident crossing edges to bottom half-space edges. Couple the box and half-space disorders by identifying edge variables under this isometry. Consider a crossing edge $e$ whose incident vertex has tangential distance at least $R$ from every adjacent face. Since $R\leq\ell_{\bm L}/2$, the opposite face is also outside its radius-$(R-1)$ neighborhood. Condition the half-space state on the spins outside the image of $\Lambda_{\bm L}$ and denote the resulting finite-volume measure by $\nu_{\bm L,r}$. The box measure $\mu_{\bm L,r}$ and $\nu_{\bm L,r}$ are Gibbs measures on the same site set. Their one-site kernels coincide at every site except possibly the inner boundary adjacent to the faces other than the selected one: the couplings and the amplitude-$r$ fixed spins on the selected face have been identified, whereas conditioning supplies possibly different fields on the remaining faces.

Let
\begin{equation*}
b_y=\sup_\zeta
\left\|\mu_{\bm L,r,y}(\,\cdot\mid\zeta)
-\nu_{\bm L,r,y}(\,\cdot\mid\zeta)\right\|_{\TV}
\end{equation*}
be the discrepancy vector in Dobrushin's comparison theorem.  Then
$0\leq b_y\leq1$, and the geometric assumptions imply $b_y=0$ whenever $\dist(y,\operatorname{supp}q_e)\leq R-1$.

Thus the Neumann series for $(I-C)^{-1}b$ starts at length $R$. Since
$\sum_x\delta_x(q_e)=2$ for $q_e=\sigma_x\psi_y$, \cref{lem:Dobrushin}
first gives
\begin{equation*}
|\mu_{\bm L,r}(q_e)-\nu_{\bm L,r}(q_e)|
\leq\frac{2\alpha^R}{1-\alpha}.
\end{equation*}
Multiplication by the identified coupling $|K_e|\leq J_*$ and averaging
preserve this estimate.  Weakening the constant from $2J_*$ to $4J_*$ gives
the safe uniform bound used below,
\begin{equation*}
\left|\E[K_e\langle q_e\rangle_{\bm L,r}]
-\E[K_{e_0}\langle\sigma_{x_0}\rangle_{\Hh_d,r}]\right|
\leq\frac{4J_*}{1-\alpha}\alpha^R.
\end{equation*}
At most a fraction $2(d-1)R/\ell_{\bm L}$ of each face is excluded. Coordinate and reflection invariance identify the same half-space response on all $2d$ oriented faces. Since every integrand is bounded by $J_*$,
\begin{equation}\label{eq:face-error}
\sup_{r\in[0,1]}
\left|\frac1{b_{\bm L}}\sum_{e\in B_{\bm L}}\E[K_e\langle q_e\rangle_{\bm L,r}]
-\E[K_{e_0}\langle\sigma_{x_0}\rangle_{\Hh_d,r}]\right|
\leq4(d-1)J_*\frac{R}{\ell_{\bm L}}
+\frac{4J_*}{1-\alpha}\alpha^R.
\end{equation}
For $0<\alpha<1$, take $R=\lfloor c\log\ell_{\bm L}\rfloor$ with $c>0$; for $\alpha=0$, take $R=1$. Then
\begin{align*}
\frac{R}{\ell_{\bm L}}
&\longrightarrow0, \quad \alpha^R=\ell_{\bm L}^{-c|\log\alpha|+o(1)}\longrightarrow0,\\
\left|\frac{\E\Delta F_{\bm L}}{b_{\bm L}}
-\tau_{\rm Dob}(\beta)\right|
&\leq\int_0^1
\left|\frac1{b_{\bm L}}\sum_{e\in B_{\bm L}}
\E[K_e\langle q_e\rangle_{\bm L,r}]-\E[K_{e_0}\langle\sigma_{x_0}\rangle_{\Hh_d,r}]\right|\d r\\
&\leq4(d-1)J_*\frac{R}{\ell_{\bm L}}
+\frac{4J_*}{1-\alpha}\alpha^R
\longrightarrow0.
\end{align*}
The independent-copy identity \cref{eq:copy-identity}, passed to the half-space by bounded convergence, yields
\begin{align*}
0
&\leq\E[K_{e_0}\langle\sigma_{x_0}\rangle_{\Hh_d,r}]
\leq\beta rv,\\
-\frac{\beta v}{2}
&=-\int_0^1\beta rv\,\d r
\leq\tau_{\rm Dob}(\beta)
\leq 0.
\end{align*}
For a deterministic exterior pattern $\psi$, replace each crossing coupling by $\psi_eK_e$. Symmetry and independence preserve its joint law, so the same argument and constants apply.

It remains to improve bulk-order concentration to surface order. Define
\begin{equation*}
\partial_{\rm in}\Lambda_{\bm L}
=\{x\in\Lambda_{\bm L}:\dist(x,\Lambda_{\bm L}^c)=1\},
\qquad
j(a)=\dist(\{x,y\},\partial_{\rm in}\Lambda_{\bm L})
\end{equation*}
for an internal edge $a=\{x,y\}$. Resample $J_a$ to an independent value $J_a'$ and let $I_a$ be the interval with endpoints $J_a$ and $J_a'$. The derivative of the correction is
\begin{align*}
\left.\frac{\partial}{\partial J_a}\Delta F_{\bm L}
\right|_{J_a=s}
&=-\langle\sigma_x\sigma_y\rangle_{\fixb,J_a=s}
+\langle\sigma_x\sigma_y\rangle_{\free,J_a=s},\\
\left|\left.\frac{\partial}{\partial J_a}\Delta F_{\bm L}
\right|_{J_a=s}\right|
&\leq\frac{\alpha^{j(a)}}{1-\alpha}
\sum_z\delta_z(\sigma_x\sigma_y)
=\frac{4\alpha^{j(a)}}{1-\alpha},\\
c_a
&\leq |J_a-J_a'|\sup_{s\in I_a}
\left|\left.\frac{\partial}{\partial J_a}\Delta F_{\bm L}
\right|_{J_a=s}\right|
\leq\frac{8J_*}{1-\alpha}\alpha^{j(a)}.
\end{align*}
For $j(a)=0$, the discrepancy vector is supported on
$\partial_{\rm in}\Lambda_{\bm L}$ and satisfies $0\leq b_y\leq1$.
Using the full Neumann series, including its $C^0$ term, gives
\begin{equation*}
\left|\langle\sigma_x\sigma_y\rangle_{\fixb,J_a=s}
-\langle\sigma_x\sigma_y\rangle_{\free,J_a=s}\right|
\leq\sum_z\delta_z(\sigma_x\sigma_y)
\sum_{n\geq0}\alpha^n
=\frac4{1-\alpha},
\end{equation*}
so the displayed bound for $c_a$ also holds at depth zero.  A
crossing-bond oscillation is at most $2J_*$. By
\cref{lem:boundary-layer-count}, at a fixed depth there are at most
$2db_{\bm L}$ internal edges. Thus
\begin{align*}
\sum_i c_i^2
&\leq b_{\bm L}(2J_*)^2
+2db_{\bm L}\sum_{j\geq0}
\left(\frac{8J_*}{1-\alpha}\alpha^j\right)^2\\
&=b_{\bm L}\left[(2J_*)^2+
\frac{32d(2J_*)^2}{(1-\alpha)^2(1-\alpha^2)}\right]
=C_*b_{\bm L}.
\end{align*}
McDiarmid's inequality and tail integration give
\begin{align*}
\Pp(|\Delta F_{\bm L}-\E\Delta F_{\bm L}|\geq u)
&\leq2\exp\left(-\frac{2u^2}{\sum_i c_i^2}\right)
\leq2\exp\left(-\frac{2u^2}{C_*b_{\bm L}}\right),\\
\E|\Delta F_{\bm L}-\E\Delta F_{\bm L}|^p
&=p\int_0^\infty u^{p-1}
\Pp(|\Delta F_{\bm L}-\E\Delta F_{\bm L}|\geq u)\d u
\leq C_{p,*}b_{\bm L}^{p/2},\\
\E\left|\frac{\Delta F_{\bm L}-\E\Delta F_{\bm L}}
{b_{\bm L}}\right|^p
&\leq C_{p,*}b_{\bm L}^{-p/2}\longrightarrow0.
\end{align*}
Finally, with $u=\varepsilon b_{\bm L_n}$,
\begin{align*}
\sum_n\Pp\left(
|\Delta F_{\bm L_n}-\E\Delta F_{\bm L_n}|
\geq\varepsilon b_{\bm L_n}\right)
\leq2\sum_n\exp\left(-\frac{2\varepsilon^2}{C_*}b_{\bm L_n}\right)<\infty,
\end{align*}
which proves the almost-sure conclusion by Borel--Cantelli.
\end{proof}

\section{Counterexamples and seam fluctuations}

We now record why the positive results require regular geometry, boundary averaging, and a precise choice of observable. The first example is a $d\geq2$ anisotropic model; the next two temporarily work in $d=1$.

\begin{proposition}[Independent boundary-noise central limit theorem]\label{prop:chain-CLT}
Let $d\geq2$ and consider $L^{d-1}$ independent zero-field Ising chains, with one i.i.d. boundary bond $K_a$ at one end of each chain. Use the first $L^{d-1}$ variables of one infinite i.i.d. sequence as $L$ varies. Then
\begin{equation}\label{eq:chain-CLT}
\Delta F_L=-\frac1\beta\sum_{a=1}^{L^{d-1}}
\log\cosh(\beta K_a).
\end{equation}
If $W=\log\cosh(\beta K)$ has nonzero finite variance, then
\begin{align}
\frac{\Delta F_L}{L^{d-1}}
&\stackrel{a.s.}{\longrightarrow}-\frac1\beta\E W
\label{eq:chain-LLN}\\
\frac{\Delta F_L-\E\Delta F_L}{L^{(d-1)/2}}
&\Longrightarrow N\left(0,\frac1{\beta^2}\Var(W)\right).
\label{eq:chain-central-limit}
\end{align}
\end{proposition}

\begin{proof}
The endpoint marginal of every free zero-field chain is uniform. For the $a$th chain,
\begin{align*}
\frac{Z_a^{\fixb}}{Z_a^{\free}}
&=\frac12\sum_{s=\pm1}e^{\beta K_as}
=\cosh(\beta K_a),\\
\frac{Z_L^{\fixb}}{Z_L^{\free}}
&=\prod_{a=1}^{L^{d-1}}\cosh(\beta K_a),\\
\Delta F_L
&=-\frac1\beta\log\frac{Z_L^{\fixb}}{Z_L^{\free}}
=-\frac1\beta\sum_{a=1}^{L^{d-1}}W_a.
\end{align*}
The strong law and the classical central limit theorem give
\begin{align*}
\frac1{L^{d-1}}\sum_{a=1}^{L^{d-1}}W_a
&\stackrel{a.s.}{\longrightarrow}\E W,\\
\frac1{L^{(d-1)/2}}\sum_{a=1}^{L^{d-1}}(W_a-\E W)
&\Longrightarrow N(0,\Var(W)).
\end{align*}
These are independent local boundary fluctuations, not a macroscopic domain wall.
\end{proof}

\begin{theorem}[One-dimensional boundary counterexamples]\label{thm:one-dimensional}
There is a bounded i.i.d. symmetric nearest-neighbor coupling law and a nested van Hove exhaustion of $\Z$ for which the expected free-to-fixed correction per crossing edge has two distinct subsequential limits. There is also a bounded i.i.d. law for which, along centered intervals in one common environment, the expected correction per crossing edge converges but the sample correction per crossing edge does not converge in probability.
\end{theorem}

The exhaustion used in the first assertion is summarized in \cref{fig:van-hove-1d}.

\begin{proof}
Take i.i.d. couplings with values $\pm K$, fix the exterior spins plus, and put $x=\beta K>0$. For an interval of $n$ sites, the high-temperature expansion gives
\begin{align}
Z_n^{+}
&=2^n\prod_{i=0}^n\cosh(\beta K_i)
\left(1+\prod_{i=0}^nT_i\right),
\notag\\
Z_n^{\free}
&=2^n\prod_{i=1}^{n-1}\cosh(\beta K_i),
\notag\\
\log\frac{Z_n^{+}}{Z_n^{\free}}
&=X_0+X_n+\log\left(1+\prod_{i=0}^nT_i\right),
\label{eq:chain-ratio}
\end{align}
where $X_i=\log\cosh(\beta K_i), T_i=\tanh(\beta K_i)$.\\
With $t=\tanh x<1$,
\begin{align*}
\left|\log\left(1+\prod_{i=0}^nT_i\right)\right|
&\leq\frac{t^{n+1}}{1-t^{n+1}}
\longrightarrow0,\\
\kappa_{\rm int}
&=\lim_{n\to\infty}\frac12\E\log\frac{Z_n^+}{Z_n^{\free}}
=\log\cosh x.
\end{align*}
For one isolated spin,
\begin{align*}
\frac{Z_{\rm iso}^{+}}{Z_{\rm iso}^{\free}}
&=\cosh(\beta(K_1+K_2)),\\
\kappa_{\rm iso}
&=\frac12\E\log\cosh(\beta(K_1+K_2))
=\frac14\log\cosh(2x).
\end{align*}
The constants differ because
\begin{align*}
4\kappa_{\rm int}-4\kappa_{\rm iso}
&=\log\frac{\cosh^4x}{\cosh(2x)}>0,
&\cosh^4x-\cosh(2x)
&=(\cosh^2x-1)^2>0.
\end{align*}

Choose intervals $I_j=[-N_j,N_j]\cap\Z$. After $I_j$, adjoin
$M_j=\lfloor\sqrt{|I_j|}\rfloor$ isolated sites, mutually separated by three lattice units, to obtain $A_j$. Choose $N_{j+1}$ so that $I_{j+1}$ contains $A_j$. Then
$I_1\subset A_1\subset I_2\subset A_2\subset\cdots$ exhausts $\Z$. Moreover,
\begin{align*}
|A_j|&=|I_j|+M_j, \quad|B(I_j)|=2,\quad |B(A_j)|=2+2M_j,\\
\frac{|B(A_j)|}{|A_j|}&=\frac{2+2M_j}{|I_j|+M_j}\longrightarrow0,\quad \frac{M_j}{|I_j|}\longrightarrow 0, (j\longrightarrow +\infty). 
\end{align*}
The same estimate holds for every fixed-width vertex boundary. Factorization over the interval component and the isolated sites gives
\begin{align*}
\frac{\E\Delta F_{I_j}}{|B(I_j)|}
&\longrightarrow-\frac{\kappa_{\rm int}}{\beta}, (j\longrightarrow +\infty).\\
\frac{\E\Delta F_{A_j}}{|B(A_j)|}
&=-\frac1\beta
\frac{2\kappa_{\rm int}+2M_j\kappa_{\rm iso}+o(1)}{2+2M_j}
\longrightarrow-\frac{\kappa_{\rm iso}}{\beta}. (j\longrightarrow +\infty).
\end{align*}
Taking $x$ small places the example inside the one-dimensional Dobrushin regime.

For the second assertion, let $|K|$ take two bounded positive values so that $X=\log\cosh(\beta K)$ is nondegenerate. On $[-n,n]$, \cref{eq:chain-ratio} gives, uniformly in the environment,
\begin{align*}
\frac{\Delta F_{[-n,n]}}{|B([-n,n])|}
=-\frac1{2\beta}(X_{-n-1}+X_n)+o(1)
=:Y_n+o(1).
\end{align*}
The variables $(Y_n)$ are i.i.d. and nondegenerate. For independent copies $Y,Y'$, choose $\varepsilon>0$ with
$p_\varepsilon:=\Pp(|Y-Y'|>3\varepsilon)>0$. For distinct sufficiently large $n,m$,
\begin{align*}
\Pp\left(\left|
\frac{\Delta F_{[-n,n]}}2-
\frac{\Delta F_{[-m,m]}}2\right|>\varepsilon\right)
\geq\Pp(|Y_n-Y_m|>3\varepsilon)
=p_\varepsilon>0.
\end{align*}
Thus the sequence is not Cauchy in probability, although its expectations converge. The two-edge boundary provides no sample averaging.
\end{proof}

\begin{figure}[t]
\centering
\includegraphics[width=.97\linewidth]{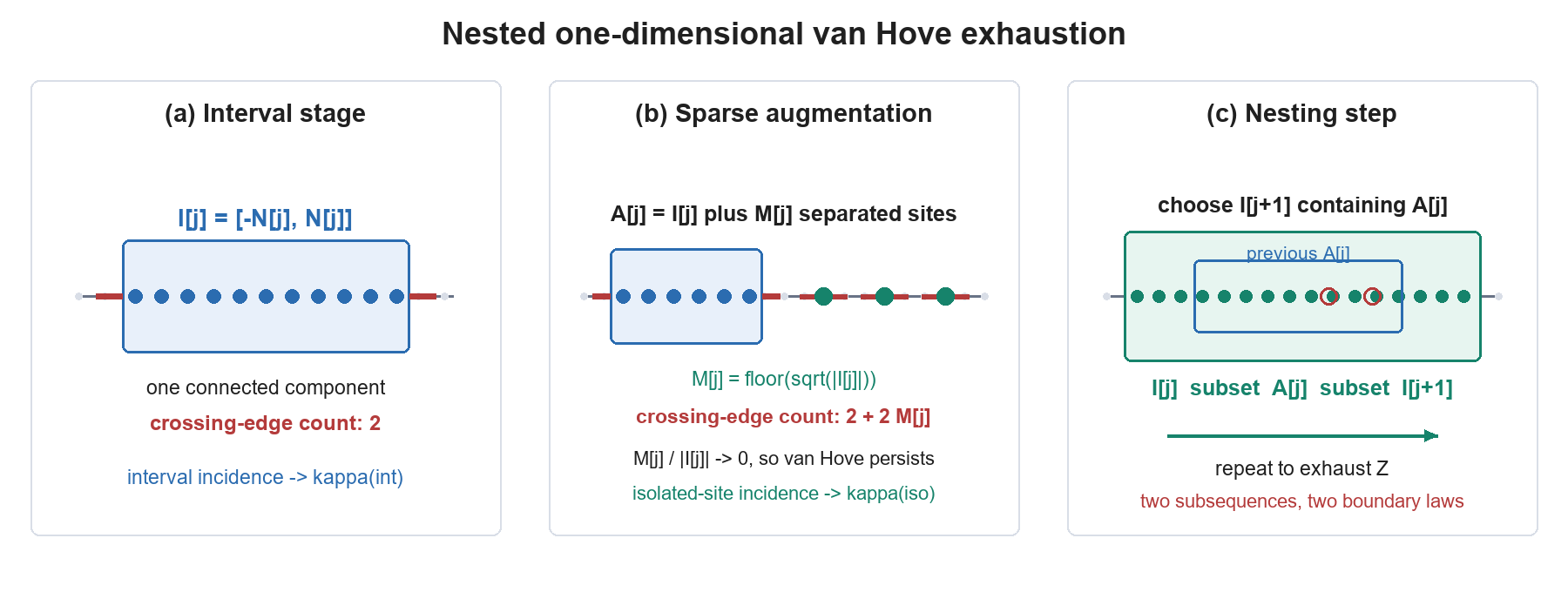}
\caption{The nested one-dimensional exhaustion in \cref{thm:one-dimensional}. An interval $I_j$ has two crossing edges. The augmented set $A_j$ adds $M_j=\lfloor\sqrt{|I_j|}\rfloor$ mutually separated isolated sites, so $M_j/|I_j|\to0$ while the isolated sites dominate the crossing-edge count. Choosing $I_{j+1}\supset A_j$ produces a nested van Hove exhaustion whose interval and augmented subsequences sample different boundary-incidence laws.}
\label{fig:van-hove-1d}
\end{figure}

The mechanism is microscopic boundary geometry rather than frustration. In a factorized model with no internal interactions, if $B(\Lambda)$ is the set of crossing edges and $\psi$ is the exterior pattern, then exactly
\begin{equation}\label{eq:factorized-boundary}
\Delta F_\Lambda=-\frac1\beta\sum_{x\in\Lambda}
\log\cosh\left(\beta\sum_{e=(x,y)\in B(\Lambda)}K_e\psi_y\right).
\end{equation}
On a flat face, one crossing edge meets almost every boundary vertex. A microscopic staircase or a dilute collection of isolated vertices can have a positive boundary density of sites incident to two or more crossing edges. Normalization by $|B(\Lambda)|$ then averages different local laws. The van Hove condition controls boundary size relative to volume but does not determine this incidence distribution.

\begin{proposition}[Finite-temperature ring twist]\label{prop:ring-twist}
For a one-dimensional ring, let $D_N=F_N^{\rm aper}-F_N^{\rm per}$ and
$Q_N=\prod_{i=1}^N\tanh(\beta J_i)$. At every finite $\beta$,
\begin{equation}\label{eq:ring-twist}
D_N=\frac1\beta\log\frac{1+Q_N}{1-Q_N}
=\frac2\beta\operatorname{artanh}Q_N.
\end{equation}
If $\E|\log|\tanh(\beta J)||<\infty$ and
$\E\log|\tanh(\beta J)|<0$, then $|D_N|$ decays exponentially almost surely. If $\Pp(J=0)>0$, then $D_N$ is eventually zero almost surely.
\end{proposition}

\begin{proof}
The one-dimensional high-temperature expansion gives
\begin{align*}
Z_N^{\rm per}=C_N(1+Q_N),
\qquad
Z_N^{\rm aper}=C_N(1-Q_N),\\
D_N
=-\frac1\beta\log Z_N^{\rm aper}
+\frac1\beta\log Z_N^{\rm per}
=\frac1\beta\log\frac{1+Q_N}{1-Q_N}.
\end{align*}
Put $\lambda=\E\log|\tanh(\beta J)|<0$. The strong law and $\operatorname{artanh}q/q\to1$ give
\begin{align*}
\frac1N\log|Q_N|
&=\frac1N\sum_{i=1}^N\log|\tanh(\beta J_i)|
\stackrel{a.s.}{\longrightarrow}\lambda,\\
\frac1N\log|D_N|
&=\frac1N\log|Q_N|+o(1)
\stackrel{a.s.}{\longrightarrow}\lambda<0.
\end{align*}
If $J_i=0$ for some $i$, then $Q_N=0$ and $D_N=0$ for every $N\geq i$.
\end{proof}

For a seam $S$, let $\omega$ denote all non-seam disorder and let $y=(y_e)_{e\in S}$ be the seam variables. If $T$ changes every $y_e$ to $-y_e$, define the same-disorder seam-flip difference
\begin{equation}\label{eq:seam-D}
D(\omega,y)=F(\omega,Ty)-F(\omega,y).
\end{equation}

\begin{theorem}[Seam variance bound]\label{thm:seam}
Conditionally on the non-seam disorder $\omega$, assume that the coordinates $(y_e)_{e\in S}$ are independent, have conditional variance equal to the deterministic value $v$ almost surely, and have a joint conditional law invariant under the simultaneous sign flip $T$. Then
\begin{equation}\label{eq:seam-bounds}
\E[D\mid\omega]=0,
\qquad
\E D^2\leq4v|S|,
\qquad
|D|\leq2\sum_{e\in S}|y_e|.
\end{equation}
If, conditionally on $\omega$, the seam coordinates are independent centered Gaussians of variance $v$, then
\begin{equation}\label{eq:seam-tail}
\Pp(|D|\geq u)\leq2\exp\left(-\frac{u^2}{8v|S|}\right).
\end{equation}
The bounds are uniform in $\beta$ and also hold for ground-state energy differences. If $|S|=L^{d-1+o(1)}$ and $(\E D^2)^{1/2}=L^{\theta_{\rm rms}+o(1)}$, then $\theta_{\rm rms}\leq(d-1)/2$. No matching lower bound or limiting distribution is asserted.
\end{theorem}

\begin{proof}
Sign-flip invariance and $D(\omega,Ty)=-D(\omega,y)$ give
\begin{align*}
\E[D\mid\omega]
=\E[D(\omega,Ty)\mid\omega]
=-\E[D(\omega,y)\mid\omega]
=0.
\end{align*}
If $y^{(e)}$ replaces $y_e$ by an independent conditional copy $y_e'$, the one-bond comparison gives
\begin{align*}
&|D(\omega,y)-D(\omega,y^{(e)})|
\leq2|y_e-y_e'|,\\
\E D^2
&=\E\Var(D\mid\omega)
\leq\frac12\sum_{e\in S}
\E\left[(D(y)-D(y^{(e)}))^2\right]\\
&\leq2\sum_{e\in S}\E(y_e-y_e')^2
=4v|S|.
\end{align*}
The deterministic comparison and Gaussian Lipschitz bound are
\begin{align*}
|D(\omega,y)|
=|F(\omega,Ty)-F(\omega,y)|
&\leq\sup_\sigma|H(\omega,Ty;\sigma)-H(\omega,y;\sigma)|
\leq2\sum_{e\in S}|y_e|,\\
\|\nabla_yD\|_2^2
&\leq\sum_{e\in S}2^2=4|S|,\\
\Pp(|D|\geq u\mid\omega)
&\leq2\exp\left(-\frac{u^2}{2v\cdot4|S|}\right).
\end{align*}
Averaging over $\omega$ proves \cref{eq:seam-tail}; the same comparison applies to ground-state energies. Finally,
\begin{align*}
(\E D^2)^{1/2}
&\leq2\sqrt{v|S|}
=L^{(d-1)/2+o(1)},
\end{align*}
so $\theta_{\rm rms}\leq(d-1)/2$. Related bounds appear in \cite{ArguinNewmanSteinWehr2014}.
\end{proof}

\begin{corollary}[Stiffness nonuniversality]\label{cor:stiffness}
For a one-dimensional ring at zero temperature,
\begin{equation}\label{eq:ring-minimum}
|D_N|=2\min_{1\leq i\leq N}|J_i|.
\end{equation}
Thus i.i.d. Rademacher couplings have root-mean-square stiffness exponent $0$, whereas symmetric unit-rate Laplace couplings have exponent $-1$. More generally, if $\Pp(|J|\leq x)\sim cx^a$ as $x\downarrow0$, the scale is $N^{-1/a}$ in probability; the same root-mean-square exponent holds under the corresponding uniform-integrability condition. Symmetry and finite moments do not determine a universal stiffness exponent.
\end{corollary}

\begin{proof}
A seam flip changes the parity of the number of unsatisfied bonds. Writing $M_N=\min_{1\leq i\leq N}|J_i|$, the two parity sectors differ by
\begin{align*}
|D_N|
=\left|\left(-\sum_{i=1}^N|J_i|+2M_N\right)
-\left(-\sum_{i=1}^N|J_i|\right)\right|
=2M_N.
\end{align*}
For Rademacher couplings, $M_N=1$. For the symmetric unit-rate Laplace law,
\begin{align*}
\Pp(M_N>t)
=\prod_{i=1}^N\Pp(|J_i|>t)
=e^{-Nt},\\
\E M_N^2
=\frac{2}{N^2},\qquad
(\E D_N^2)^{1/2}
=\frac{2\sqrt2}{N}.
\end{align*}
If $\Pp(|J|\leq x)\sim cx^a$, then
\begin{align*}
\Pp(N^{1/a}M_N>t)
=\left(1-\Pp(|J|\leq tN^{-1/a})\right)^N
\longrightarrow e^{-ct^a},
\quad t>0.
\end{align*}
The root-mean-square scale is the same under uniform integrability \cite{Resnick1987}.
\end{proof}

\section{Scope and the Gaussian low-temperature problem}

The unconditional fixed-depth theorem survives at $\beta=\infty$ because its proof uses only cutting and Lipschitz comparison. The derivative representation, DLR response, and BRM arguments are finite-temperature statements. Passing \cref{eq:gaussian-delta} or \cref{eq:integrated-response} to zero temperature would require additional estimates uniform in $\beta$ and is not done here.

Nor does one-face BRM assemble the simultaneous correction from all faces of a cube. Such a theorem would require oriented estimates uniform under the boundary conditions created by a face-by-face telescoping, together with control of edges and corners. The all-face conclusion in \cref{thm:Dobrushin-cube} is available because Dobrushin comparison supplies that stronger uniformity. The one-dimensional example in \cref{thm:one-dimensional} also shows why no arbitrary-shape conclusion should be inferred from a rectangular theorem.

\subsection{Comparison of methods}

\begin{center}
\small
\begin{tabularx}{\textwidth}{@{}>{\raggedright\arraybackslash}p{.22\textwidth}>{\raggedright\arraybackslash}p{.36\textwidth}X@{}}
\toprule
Method & Established here & Missing input for the Gaussian low-temperature model \\
\midrule
Cutting and almost-additivity & Rectangular tangential pressure at fixed depth & Control as the normal depth diverges \\
Gaussian or copy interpolation & Finite-volume response identities and sign bounds & A full-sequence limit of the boundary response \\
Generic concentration & Upper bounds for centered fluctuations & Convergence of normalized expectations; surface-order control in $d=2$ \\
Response mixing & Selector-independent weighted one-spin response & Verification, or a weaker substitute, in the ordinary model \\
Dobrushin comparison & Regular-rectangle limit and surface-order concentration & Bounded interactions and $\alpha<1$ \\
Metastate compactness & Subsequential DLR states \cite{AizenmanWehr1990,NewmanStein1997} & Independence of the induced boundary response \\
\bottomrule
\end{tabularx}
\end{center}

\subsection{What stochastic localization would have to supply}

For a finite-volume Gibbs measure $\mu_0$ on $\{-1,1\}^{\Lambda}$, stochastic localization introduces a measure-valued martingale of the form \cite{Eldan2013}
\begin{equation}\label{eq:stochastic-localization}
\d\mu_s(\sigma)=\mu_s(\sigma)
\langle\sigma-a_s,\d B_s\rangle,
\qquad a_s=\E_{\mu_s}\sigma,
\qquad A_s=\Cov_{\mu_s}(\sigma),
\end{equation}
so that $\d a_s=A_s\d B_s$. To prove the normal localization needed here, one would require a bound, uniform in volume, interpolation amplitude, and disorder, on the projection of $A_s$ from a remote face to the boundary observable. The estimate must also be stable under the boundary kernels produced by DLR disintegration. The finite-volume martingale representation alone gives neither property, and the present work does not assume a log-concavity or spectral-independence hypothesis from which they would follow.

For the ordinary fully occupied Edwards--Anderson model, all nearest-neighbor couplings are i.i.d. Gaussians and are nonzero almost surely. At low temperature the bounded Dobrushin theorem is unavailable, and
$p_\beta=\E\tanh(\beta|J|)\to1$, so the subcritical disagreement criterion fails for every $d\geq2$ at sufficiently large $\beta$. Compactness and metastate constructions give subsequential states \cite{AizenmanWehr1990,NewmanStein1997}, while \cref{prop:finite-calculus} gives the finite-volume Gaussian identity; neither supplies the missing normal localization. This issue is separate from the mean-field picture \cite{SherringtonKirkpatrick1975,Parisi1979,AizenmanLebowitzRuelle1987} and from positive surface tension in unfrustrated contour regimes \cite{DKS1992}.

\begin{problem}[Ordinary fully occupied low-temperature Gaussian EA model]\label{prob:Gaussian}
Determine, separately by dimension and at fixed finite inverse temperature, whether BRM holds for the ordinary fully occupied Gaussian Edwards--Anderson half-space specifications. If it fails, identify a weaker condition that still makes the weighted boundary response independent of the DLR selector and controls the normal-depth limit. Determine under such a condition whether a full-sequence regular-cube free-to-fixed surface limit follows. At zero temperature, a response formula and the interchange of the volume and temperature limits require separate arguments. These questions are not resolved by the present paper; no broader claim about their status in the literature is made here.
\end{problem}

\section*{Declarations}

\subsection*{Funding}
The authors declare that no funds, grants, or other support were received during the preparation of this manuscript.

\subsection*{Competing interests}
The authors have no relevant financial or non-financial interests to disclose.

\subsection*{Authors' contributions}
All authors contributed to the conceptualization, mathematical analysis, proof verification, and writing of the manuscript. All authors contributed equally to this work. All authors read and approved the final manuscript.

\subsection*{Data availability}
Data sharing is not applicable because no datasets were generated or analyzed in this theoretical study.

\subsection*{Ethical Approval/Human and Animal Rights}
This article does not contain any studies with human participants or animals performed by any of the authors.

\Addresses
\end{document}